\documentclass[sigplan,nonacm,10pt]{acmart}
\AtBeginDocument{%
  }

\setcopyright{none}
\renewcommand{\footnotetextcopyrightpermission}[1]{}

\usepackage{xspace}
\usepackage{pifont}
\usepackage{fontawesome5}
\usepackage{array}
\usepackage{tcolorbox}
\usepackage{subcaption}
\usepackage{amsthm}
\usepackage{etoolbox}
\usepackage{listings}
\usepackage{multirow}
\usepackage{enumitem}
\usepackage{placeins}

\AtBeginEnvironment{tcolorbox}{\small\setlength{\itemsep}{1pt}\setlength{\parsep}{0pt}\setlength{\topsep}{0pt}}

\newcommand{\para}[1]{\smallskip\noindent\textbf{#1}}

\theoremstyle{definition}

\newtheorem{theorem}{Theorem}
\newtheorem{claim}{Claim}
\newtheorem{fact}{Fact}

\newif\ifspace

\def\Snospace~{\S{}}

\newcommand{\sys}{\textsc{Chora}\xspace}
\newcommand{\runtime}{\textsc{Chora}\xspace}
\newcommand{\typeonebubble}{propagation bubble\xspace}
\newcommand{\Typeonebubble}{Propagation bubble\xspace}
\newcommand{\typetwobubble}{jitter bubble\xspace}
\newcommand{\Typetwobubble}{Jitter bubble\xspace}

\newcommand{\define}[1]{\textit{#1}\xspace}

\newcommand{\var}[1]{\textit{#1}\xspace}
\newcommand{\msg}[1]{{\small \textsc{#1}}\xspace}

\newcommand{\cmd}[1]{{\small \textsc{#1}}\xspace}

\newcommand{\cmt}[1]{}

\newcommand{\envicon}[2][black]{{\color{#1}\faEnvelope[regular]}\,{\small #2}}

\newif\ifshowcomments
\showcommentstrue

\ifshowcomments
  \newcommand{\editorial}[2]{{\textcolor{#1}{#2\xspace}}}
  \newcommand{\cy}[1]{\editorial{cyan}{chaoyi: {#1}}}
  \newcommand{\wyl}[1]{\editorial{orange}{YL: {#1}}}
  \newcommand{\kira}[1]{}
  \newcommand{\lijl}[1]{\editorial{cyan}{JL: {#1}}}
  \newcommand{\inho}[1]{\editorial{brown}{inho: {#1}}}
\else
  \newcommand{\editorial}[2]{}
  \newcommand{\cy}[1]{}
  \newcommand{\wyl}[1]{}
  \newcommand{\lijl}[1]{}
  \newcommand{\inho}[1]{}
  \newcommand{\kira}[1]{}
\fi

\begin{document}

\title[]{
Reducing Pipeline Bubbles in \\
Message-Efficient State Machine Replication
}

\author{Yiliang Wan\textsuperscript{1},
Inho Choi\textsuperscript{1},
Nitin Shivaraman\textsuperscript{2},
Akshaye Shenoi\textsuperscript{3}, \\
Xiang Liu\textsuperscript{1},
Tao Luo\textsuperscript{2},
Haifeng Sun\textsuperscript{1},
Chaoyi Ruan\textsuperscript{1*},
Jialin Li\textsuperscript{1}}
\affiliation{%
  \vspace{8pt}
  \institution{\textsuperscript{1}National University of Singapore \qquad \textsuperscript{2}IAIC, A*Star Singapore \qquad \textsuperscript{3}ETH Zurich}
  \country{ }
}
\renewcommand{\authors}{Yiliang Wan, Inho Choi, Nitin Shivaraman, Akshaye Shenoi, Xiang Liu, Tao Luo, Haifeng Sun, Chaoyi Ruan, Jialin Li}
\makeatletter
\g@addto@macro\@authornotes{\stepcounter{footnote}\footnotetext{Corresponding author.}}
\makeatother

\renewcommand{\shortauthors}{Wan et al.}

\begin{abstract}

Round-based pipelining in state machine replication (SMR) protocols reduces amortized bottleneck message complexity to $O(1)$, but waiting for network propagation and accommodating timing variation create pipeline bubbles that can erase this benefit.
We present \sys, a message-efficient SMR protocol co-designed with a network runtime to minimize these bubbles under common-case synchrony.
\sys introduces a new round construct that decouples propagation delay from round length.
It then builds a carefully engineered runtime that minimizes variation in processing jitter.
The co-designed protocol leverages strong synchrony to achieve $O(1)$ amortized bottleneck message complexity with high resource utilization, and maintains safety and liveness when timing predictability degrades.
Evaluation shows that \sys reaches 3.92 million operations per second with five RDMA replicas, achieving 2.87$\times$ throughput of the state-of-the-art, while maintaining comparable latency.
Even on this single-hop RDMA network that favors conventional rounds, \sys reduces time lost to pipeline bubbles, leaving more CPU time for protocol processing.
Across three to seven replicas, the CPU utilization rate for protocol processing increases by 211\%--421\% and throughput improves by $5.3$--$6.6\times$ over conventional rounds.

\end{abstract}

\maketitle

\section{Introduction}
\label{sec:introduction}

State machine replication (SMR)~\cite{smr} is the standard technique for building fault-tolerant services:
replicas agree on a common sequence of client operations and execute them in the same order, preserving linearizability~\cite{linearizability1990} even when some replicas fail.
SMR underlies widely deployed systems such as Chubby~\cite{chubby2006}, ZooKeeper~\cite{zookeeper2010}, and Spanner~\cite{spanner2013}.
Because replication sits on the critical path of these systems, the throughput of the SMR protocol often limits the capacity of the services built on it.

SMR protocols usually improve throughput by \emph{pipelining}, overlapping multiple consensus instances to hide communication delay.
Once enough proposals are in flight, however, throughput is bounded by the \emph{bottleneck message complexity}: the number of messages the busiest replica must process per committed proposal.
In conventional leader-based protocols~\cite{paxosms2001,raft2014}, the leader handles $O(N)$ messages per proposal, where $N$ is the replica count.
Multi-leader protocols~\cite{mencius2008,epaxos} distribute proposal work across replicas, improving utilization and load balancing, but each proposer still handles $O(N)$ messages for its own proposals.
\emph{Aggregation} reduces this cost further by letting a single message carry work for many instances:
a replica's outgoing proposal can also carry a cumulative acknowledgement for proposals it has already received from other replicas.
The effectiveness of aggregation, however, depends on transmission timing: responding immediately to each arrival aggregates little, while waiting for more progress adds latency and may stall on a delayed peer.
Without predictable message arrival and processing times, it is difficult to coordinate this pipeline while keeping both waiting and message-processing overhead low.

Round-based pipelining makes replication \emph{message-efficient} by coordinating aggregation to reduce the number of messages replicas send and process per committed proposal.
In synchronous \emph{rounds}~\cite{lynch1996distributed,barnoy1991}, each replica multicasts one message per round and processes that round's peer messages before advancing.
Every replica receives all $O(N)$ proposals of a round and can piggyback $O(N)$ acknowledgements on its next outgoing message.
This reduces the amortized bottleneck message complexity to $O(1)$.
The round length depends on the synchrony bound which includes the maximum propagation and processing delays.
Our insight is that the approach trades the message bottleneck for timing: throughput now depends on how tightly rounds can be sized physically. 

This approach poses two practical challenges.
First, timing bounds can be violated in practice, so protocol correctness cannot depend on synchrony guarantees;
this is why practical SMR protocols assume only partial synchrony~\cite{partialsync1988,paxosms2001,raft2014,vr2012}, where delays are bounded during sufficiently stable periods of execution.
Second, even when a bound holds, conventional rounds waste much of the available processing capacity:
a round must cover both network propagation and protocol processing, and because propagation dominates in typical deployments, replicas spend much of each round idle waiting for messages to arrive---a \emph{\typeonebubble}.
Additional slack reserved for variation in delivery and processing times creates a \emph{\typetwobubble}.
Measurements (\S\ref{sec:case}) show that these bubbles can erase pipelining’s throughput benefit.

We present \sys, a \emph{message-efficient} SMR protocol that achieves \emph{zero-bubble pipelining} in the common case through co-design with a network runtime.
Zero-bubble execution requires timing that is tightly predictable rather than merely bounded; the runtime provides this \emph{strong synchrony}.
Because the two bubbles arise from different sources, one from the round's structure and one from timing variation, the design consists of two components, one for each bubble.

First, the \sys \emph{protocol} (\S\ref{sec:protocol}) makes correctness independent of synchrony and removes the \typeonebubble.
In the common case, each replica multicasts one \emph{pulse} per round but processes available messages from earlier rounds,
overlapping network propagation with useful work so the slowest link no longer dictates round length.
Each pulse combines a new proposal with a cumulative acknowledgement for earlier proposals, achieving $O(1)$ amortized bottleneck message complexity;
we call this common-case execution the \emph{pulsing mode}.
When timing predictability degrades, the same protocol state and commit rules drive a \emph{responsive mode} that transmits protocol messages immediately, so safety and liveness never depend on synchrony.

Second, the \emph{runtime} (\S\ref{sec:runtime}) realizes strong synchrony in practice, so that the rounds the protocol defines logically remain short in physical time.
It shrinks the \typetwobubble by isolating round processing on a dedicated core, separating application handling and packet I/O through FIFO queues and kernel-bypass networking.
Adaptive round boundaries (a minimum duration and a soft deadline with early close) let each round end as soon as its expected work completes, without budgeting every round for rare tail delays.
A controller tunes these boundaries online and switches the group between pulsing and responsive execution when synchrony degrades.

Together, the two components remove both bubbles from common-case execution while retaining correctness under partial synchrony.
We evaluate \sys with RDMA and UDP transports on a local cluster and on AWS EC2. 
With five RDMA replicas, \sys reaches 3.92\,Mops/s, 2.87$\times$ Mu~\cite{mu}, and its advantage persists in intra-AZ and cross-AZ cloud deployments.
An ablation shows that propagation-decoupled rounds reduce time lost to pipeline bubbles, leaving more CPU time for protocol processing.
Across three to seven replicas, the CPU utilization rate for protocol processing increases by 211\%--421\% and throughput improves by $5.3$--$6.6\times$ over conventional rounds.
Finally, pulsing generalizes beyond \sys: adapting the pulsing mode to Derecho~\cite{derecho,spindle} and Hermes~\cite{hermes} improves their throughput by 41.6\% and 28.1--41.8\%, respectively. 

\section{Background}
\label{sec:background}

\noindent\textbf{State machine replication (SMR)}~\cite{smr} replicates a deterministic state machine across a set of servers.
Replicas agree on a common sequence of operations and execute them in the same order,
providing linearizability~\cite{linearizability1990} even when a subset of replicas fails.
In this work, we assume all replicas follow the protocol, and can only fail by crashing.
SMR protocols have been widely deployed to
provide strong fault tolerance guarantees in distributed systems~\cite{chubby2006,spanner2013,zookeeper2010,gaios}.

\noindent\textbf{Replicated logs} organize operations in practical SMR protocols~\cite{paxosms2001,raft2014,vr2012,mencius2008}.
Replicas coordinate to assign proposals to log slots,
and a proposal becomes committed once the replicas reach agreement on its position in the log.
A correct SMR protocol provides two key guarantees:
\emph{safety} prevents conflicting proposals from being committed at the same log slot,
while \emph{liveness} ensures that submitted proposals eventually get committed.

\para{Synchrony} affects liveness and performance; safety must hold without timing guarantees.
Consensus cannot guarantee progress in a fully asynchronous system with failures~\cite{flp},
so practical protocols assume partial synchrony~\cite{partialsync1988,paxosms2001,pmmc2015,vr2012,raft2014},
under which message delivery and processing are eventually bounded, and rely on it for liveness.
Synchrony also shapes performance:
message delay bounds commit latency,
and protocols with stronger timing guarantees, such as synchronized clocks or bounded delays,
exploit them to order operations with less coordination~\cite{chubby2006,gfs,spanner2013,clockrsm2014,nezha,tempo2021}.

\para{Proposal and acknowledgement} form the two stages of a leader-based commit~\cite{paxosms2001,vr2012,raft2014}.
In the \emph{proposal} stage, the leader assigns a client command to the next log slot and sends the proposal to all replicas.
In the \emph{acknowledgement} stage, each replica records the proposal and sends an acknowledgement back to the leader,
which commits the proposal once it has received acknowledgements from a quorum.
Each stage involves a network traversal and message processing at both ends,
so a proposal takes at least one round trip to commit.

\para{SMR pipelining} overlaps work across log slots to improve resource utilization and throughput.
Committing one proposal at a time would leave replicas idle for most of the round trip.
The leader starts the next slot's proposal stage before earlier slots finish their acknowledgement stage, keeping multiple consensus instances in flight.
At any moment, the leader is sending proposals for new slots while receiving acknowledgements for earlier ones,
and each replica alternates between processing incoming proposals and sending acknowledgements.
Multi-phase BFT protocols pipeline in the same way, overlapping protocol stages across successive proposals or blocks~\cite{hotstuff2019,narwhal-tusk}.

\section{Motivation and Challenges}
\label{sec:case}

\begin{figure*}[t]
\centering
\includegraphics[width=0.93\linewidth]{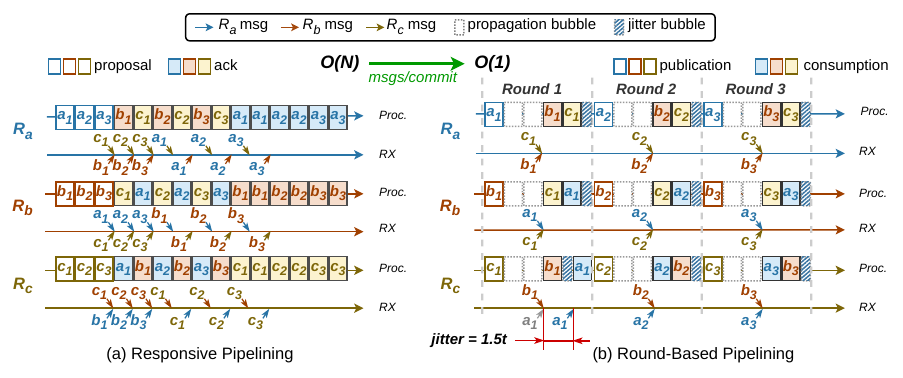}

\vspace{-15pt}
\caption{
Round-based pipelining reduces amortized bottleneck message complexity from $O(N)$ in (a) to $O(1)$ in (b).
Each block costs $t$, and propagation takes $\delta=2t$.
}
\label{fig:timeline_comparison}
\end{figure*}



SMR protocols pipeline consensus instances to overlap their communication and processing.
This section examines what limits such a pipeline and how synchrony helps.
Responsive pipelining is bounded by its bottleneck message complexity,
and round-based pipelining removes this bound by aggregating proposals and acknowledgements across instances (\autoref{sec:case:rounds}).
Rounds, however, leave two kinds of idle time, or \emph{bubbles}:
a \emph{\typeonebubble} from placing network propagation inside the round barrier,
and a \emph{\typetwobubble} from variation in message delivery and processing times (\autoref{sec:case:bubbles}).
We then analyze how each bubble bounds pipeline performance (\autoref{sec:case:analysis}).



\subsection{From Responsive to Round-Based Pipelining}
\label{sec:case:rounds}





\noindent\textbf{Responsive pipelining and its bottleneck.}
As discussed in \autoref{sec:background}, practical SMR protocols pipeline multiple consensus instances
to hide communication latency and improve resource utilization.
Conventional pipelining follows a \emph{responsive} execution model:
replicas process protocol messages as they arrive and immediately perform the corresponding protocol actions.
Different consensus instances thus progress independently,
without coordinating their transmissions through common round boundaries.
Responsive pipelining hides end-to-end communication latency but does not reduce the protocol work required for each proposal.
Once the pipeline is saturated, throughput is determined by its \emph{bottleneck message complexity},
i.e., the number of protocol messages processed by the busiest replica per committed proposal.



For leader-based protocols~\cite{pmmc2015,vr2012,raft2014}, the bottleneck is the leader,
which processes an acknowledgement from every replica for each commit
and therefore handles $O(N)$ messages per proposal, where $N$ is the replica count.
Multi-leader (or equivalently, leaderless) SMR protocols~\cite{mencius2008,epaxos,epaxos2021,smr-scalability}
address this bottleneck by letting all replicas propose.
As \autoref{fig:timeline_comparison}(a) illustrates,
this utilizes the otherwise idle compute cycles on non-leader replicas and distributes the proposer workload across the cluster.
However, each commit still requires its proposer to collect acknowledgements from a quorum,
so the bottleneck message complexity remains $O(N)$.
Moreover, the lack of a single designated leader introduces coordination overhead,
such as blocking on slow leaders in Mencius~\cite{mencius2008}
and conflict resolution for dependent operations in EPaxos~\cite{epaxos}.





\noindent\textbf{Rounds amortize bottleneck message complexity.}
Both designs have $O(N)$ bottleneck message complexity because they treat concurrent consensus instances independently.
If replicas instead coordinate their transmissions, a single protocol message can carry progress for multiple instances.
In particular, a replica can combine a new proposal with a cumulative acknowledgement for all proposals it has received into one multicast.
A conventional synchronous round~\cite{lynch1996distributed,barnoy1991} has two phases.
In the \emph{publication phase}, a replica constructs and multicasts a message; in the \emph{consumption phase}, it receives and processes peer messages sent in that round.
Round-based execution of this kind is common in DAG-based BFT protocols,
where each replica broadcasts one block per round that references the blocks it received in earlier rounds~\cite{narwhal-tusk,dagrider,bullshark}.
As \autoref{fig:timeline_comparison}(b) shows, with $N$ replicas each multicasting once per round,
each replica processes $O(N)$ messages while the pipeline advances $N$ proposals,
yielding $O(1)$ amortized bottleneck message complexity.

\subsection{Two Bubbles in Round-Based Pipelining}
\label{sec:case:bubbles}




Rounds, however, are not free.
Replicas now advance in lockstep, and any time a replica spends waiting within a round is idle time, or a \emph{bubble}, in the pipeline.
\autoref{fig:timeline_comparison}(b) marks two kinds:
a \typeonebubble (dotted) that comes from the round structure itself,
and a \typetwobubble (hatched) that comes from timing variation.





\para{\Typeonebubble.}
A conventional round requires its consumption phase to process peer messages published in the same round.
In \autoref{fig:timeline_comparison}(b), a replica multicasts and processes its own proposal at the start of a round,
but its peers' messages arrive only after the propagation delay, so it idles until they do and then processes them.
If $\delta$ is the mean one-way propagation delay and $t$ is the mean per-message processing time,
the round must therefore cover approximately $\delta + Nt$ even under ideal timing.
The idle $\delta$ is the \typeonebubble.
It grows with the longest propagation delay between any pair of replicas,
since the round cannot end before the slowest message has arrived and been processed.

The bubble stems from the same-round requirement, not from coordination itself.
Our key observation is that coordinating \emph{transmission} does not require consuming messages from the same round.
We therefore introduce \emph{propagation-decoupled rounds}: publication still multicasts once per round, while consumption processes peer messages from earlier rounds.
In \autoref{fig:decoupled_round}, each replica processes its own proposal together with the messages its peers sent in the previous round,
so propagation overlaps with useful processing and a round holds only the $N$ processing blocks.
Messages on different links may even lag by different numbers of rounds
without making the propagation delay of the slowest link part of every round (\autoref{fig:round_comparison}),
so the ideal round length reduces from approximately $\delta + Nt$ to $Nt$.
Propagation still contributes to end-to-end commit latency, but it no longer limits the steady-state pipeline rate.





\noindent\textbf{\Typetwobubble.}
Decoupling removes the mean propagation delay from the round, but a second bubble remains.
Replicas must still cross each round boundary together,
so the round length must cover not only the mean processing time but also its variation.
In \autoref{fig:timeline_comparison}(b), a peer's message may arrive up to $1.5t$ late,
so every round reserves that much slack, and a replica that finishes early idles until the boundary before issuing its next pipeline step.
Conversely, if the round is too short, some replicas fail to complete their processing within the round
and must defer unfinished work to later rounds.
Round-based pipelining thus faces a trade-off:
long rounds leave replicas idle, while short rounds increase the likelihood of missed round boundaries.
Either way, the bubble can negate the throughput benefit of reduced bottleneck message complexity and lead to worse latency.
Unlike the \typeonebubble, its size is set by how unpredictable message delivery and processing are, not by their mean.





This is where stronger common-case synchrony becomes useful.
Suppose that the variance of message delivery and processing latencies is low, and their latency distributions are known to all replicas.
With this timing information, replicas can choose a tight round bound
such that each replica completes its expected processing within the round with high probability.
When the variance is low, the round length only needs to cover the mean processing time plus a small buffer for the tail,
so most of the round consists of useful work and the \typetwobubble is small.
This property does not require equal propagation delays across links,
since propagation-decoupled rounds allow messages from different links to span different numbers of rounds.
What matters is low variation in message delivery and protocol processing.
When such timing predictability is unavailable, replicas can instead progress responsively based on message arrivals,
preserving progress without relying on tight round coordination.

\begin{figure}[t]
\centering
\includegraphics[width=\linewidth]{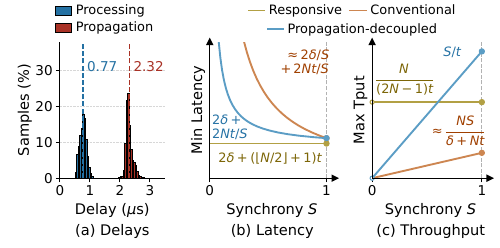}
\vspace{-21pt}
\caption{
Round delays and analytical bounds.
(a) Measured round processing and one-way propagation; dashed lines mark means.
(b--c) Optimal latency and maximum throughput versus synchrony $S$ ($\delta \gg t$).
}
\label{fig:pipeline-bubbles}
\end{figure}

\subsection{Pipeline Performance Bottleneck Analysis}
\label{sec:case:analysis}


How much each bubble costs depends on the timing of the deployment.
The \typeonebubble depends on how propagation delay compares with processing delay,
and the \typetwobubble depends on how predictable delivery and processing are.
We analyze each in turn, first with measurements of propagation and processing delays, 
then with an analytical model as functions of a synchrony coefficient.



\noindent\textbf{Propagation delay versus processing delay.}
A conventional round of length $\delta + Nt$ spends $\delta$ waiting,
so the \typeonebubble occupies a fraction $\delta / (\delta + Nt)$ of every round.
\autoref{fig:pipeline-bubbles}(a) measures both terms on our one-hop three-replica RDMA testbed at saturation.
Mean one-way propagation is 2.32\,$\mu$s, three times the 0.77\,$\mu$s a replica spends processing a round's messages,
so the \typeonebubble takes about three quarters of a conventional round even after message aggregation.
The two means approximate the measured round lengths with fixed deadlines and early exit:
$\delta + Nt = 3.09$\,$\mu$s against 3.45\,$\mu$s for conventional rounds,
and $Nt = 0.77$\,$\mu$s against 0.83\,$\mu$s for propagation-decoupled rounds.
The gap in each case is the \typetwobubble, which we quantify next.

\para{Quantifying synchrony.}
The \typetwobubble is the slack a round reserves for timing variation,
e.g., the 0.5 to 1.1\,$\mu$s spread of per-round processing in \autoref{fig:pipeline-bubbles}(a).
We quantify it with the \emph{synchrony coefficient} $S \in (0,1]$:
the fraction of a round spent on useful processing once the round is bounded to cover that variation, so $1-S$ implies the \typetwobubble.
$S=1$ means no slack is needed, and a smaller $S$ means a looser round bound.
$S^x$ denotes the coefficient for a bound at the $x$-th percentile.
\autoref{sec:network} gives the formal definition of synchrony coefficient $S$.

\autoref{fig:pipeline-bubbles}(b--c) plots the resulting bounds against $S$, modeling all slack with a common factor $1/S$.
A conventional round lasts about $(\delta+Nt)/S$, bounding throughput at $NS/(\delta+Nt)$;
a propagation-decoupled round lasts about $Nt/S$, bounding throughput at $S/t$, which reaches the processing limit $1/t$ as $S \to 1$.
Both improve with tighter synchrony, but only decoupled rounds escape $\delta$.
A responsive multi-leader pipeline, in contrast, is bounded at $N/((2N-1)t)$ regardless of $S$,
since each replica also performs follower work for the other proposers.
The model assumes a uniform per-message cost $t$; actual costs vary,
but every message incurs deserialization, validation, and state updates, so bottleneck message complexity remains a fair proxy for processing work.

\autoref{fig:pipeline-bubbles}(b) shows the other side of the trade-off.
Every design pays two propagation legs, $2\delta$, for a proposal to reach the replicas and their acknowledgements to return.
The responsive baseline adds only path processing, about $(\lfloor N/2 \rfloor+1)t$, whereas round-based pipelines add two round crossings:
$2\delta/S + 2Nt/S$ for conventional rounds, which stretch each propagation leg, and $2\delta + 2Nt/S$ for decoupled rounds.
With the delays of \autoref{fig:pipeline-bubbles}(a), that is about 1\,$\mu$s on top of 4.6\,$\mu$s of propagation, so rounds cost little latency once decoupled.


These bounds set the targets for \sys:
remove the $\delta$ term with propagation-decoupled rounds, and push $S$ toward one with a strong-synchrony runtime.
The next section outlines how \sys realizes both.

\section{\sys Overview}

\begin{figure}[t]
\centering
\includegraphics[width=.85\linewidth]{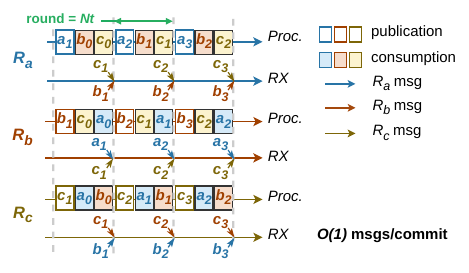}
\vspace{-12pt}
\caption{
\textbf{Zero-bubble pipelining.}
Propagation-decoupled rounds overlap processing ($t$ per block) with propagation ($\delta=2t$ one-way), shortening ideal rounds to $Nt$ while retaining $O(1)$ amortized bottleneck message complexity.
}
\label{fig:decoupled_round}
\end{figure}

To address the above bubbles, we present \sys, an SMR protocol co-designed with a
strong-synchrony network runtime that realizes 
\emph{zero-bubble} SMR pipelining.
\sys preserves the message efficiency of round-based aggregation while
eliminating idle gaps between successive pipeline stages.
\sys addresses the \typeonebubble and \typetwobubble through two techniques.
\begin{itemize}[leftmargin=*]
\item
\textbf{Propagation-decoupled pipeline (\S\ref{sec:protocol}).}
\sys removes the same-round dependency between publication and consumption to eliminate the \typeonebubble.
Unlike a conventional round, in which replicas transmit and process messages 
from the same round, a propagation-decoupled round processes messages 
transmitted in previous rounds. 
Messages for future pipeline stages propagate concurrently with the processing 
of messages that have already arrived. 
Once the pipeline is filled, network propagation is removed from the critical 
path of each round, allowing replicas to advance proposal, acknowledgment, and 
commitment stages continuously.
We design an SMR protocol that puts the propagation-decoupled pipeline into practice.
The \sys protocol uses a high-performance \emph{pulsing mode} under strong 
synchrony and a \emph{responsive fallback} 
when this timing property does not hold.

\item
\para{Strong-synchrony runtime (\S\ref{sec:runtime}).}
\sys runtime builds an execution substrate that provides strong synchrony
for short, predictable rounds and minimizes the \typetwobubble. Conventional partial
synchrony assumes only an eventual finite delay bound and suffices for protocol
liveness. Strong synchrony further requires communication and processing
delays to have low variance and a tight probabilistic bound. It is a
common-case performance property rather than an additional safety assumption.
\runtime isolates timing-sensitive protocol operations from application and
transport activity, then adapts flexible round boundaries to observed timing
and protocol progress. It monitors clock and pipeline health and switches to
responsive execution when strong synchrony degrades.
\end{itemize}
The two techniques address complementary dimensions of the pipeline. 
Propagation-decoupled rounds overlap \emph{mean} network propagation with 
protocol processing, while strong synchrony minimizes the residual waiting 
required for \emph{tail} latency variation. 
Together, they produce the common-case execution shown in \autoref{fig:decoupled_round}, 
in which proposal transmission, acknowledgment, and commitment advance 
across rounds without pipeline bubbles.

\section{\sys Protocol}
\label{sec:protocol}

To achieve zero-bubble pipelining, we design the \sys protocol 
based on the propagation-decoupled pipeline.
We first overview the protocol in \S\ref{subsec:protocol-overview}.
We then describe the primary pulsing mode in
\S\ref{subsec:pulsing-mode} and its responsive fallback in
\S\ref{subsec:responsive-mode}.
Finally, we discuss how to recover from different failure types 
in \S\ref{subsec:failure-recovery}.
\autoref{sec:protocol:detail} gives the complete message specification, transition rules, and safety proof.

\begin{figure}
\begin{tcolorbox}[title={\sys Replica Local State}, top=3pt, bottom=3pt, left=5pt, right=5pt, boxsep=2pt]
\footnotesize
\begin{itemize}[nosep,leftmargin=*]
    \item \textbf{\var{log}} - replication log
    \item \textbf{\var{cmds}} - buffered client request commands
    \item \textbf{\var{view}} - current view number
    \item \textbf{\var{role}} - current role (initiator, candidate, or follower)
    \item \textbf{\var{voted-for}} - candidate that this replica voted for
    \item \textbf{\var{voted-by}} - set of replicas that voted for this replica
    \item \textbf{\var{view-base}} - first log slot of the current view
    \item \textbf{\var{next-propose}} - next log slot to propose
    \item \textbf{\var{last-append}} - log slot before the first empty slot
    \item \textbf{\var{last-ack}} - last acknowledged slot in the current view
    \item \textbf{\var{last-commit}} - last committed log slot
    \item \textbf{\var{acked}} - acknowledgement indices, one per replica
\end{itemize}
\end{tcolorbox}
\vspace{-14pt}
\caption{Replica state in \sys.}
\label{fig:replica_state}
\end{figure}

\begin{figure*}[t]
\centering
\includegraphics[width=\linewidth]{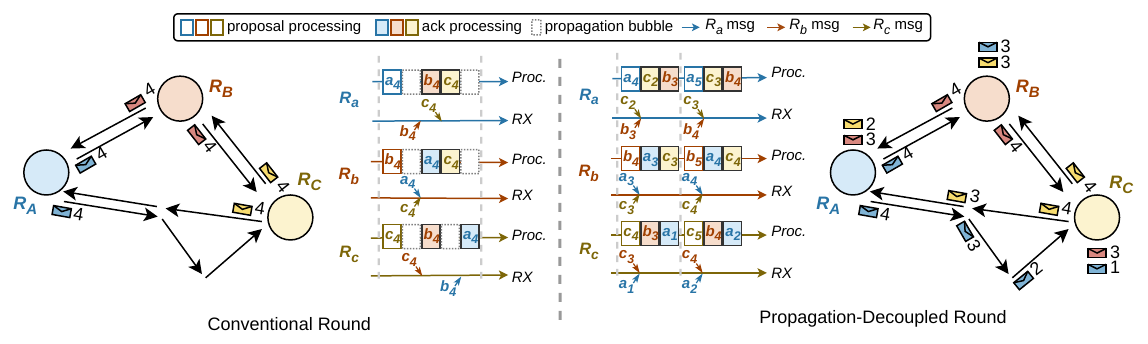}
\vspace{-6pt}
\caption{
Propagation-decoupled rounds process messages from different earlier rounds, keeping the propagation delay out.
Envelope labels give publication rounds (e.g., \envicon{3} denotes a pulse published in round 3).
Each black arrow represents the unit time of propagation delay ($\Delta \delta = t$, where $t$ is per-message processing time).
}
\label{fig:round_comparison}
\label{fig:network}
\end{figure*}

\subsection{Protocol Overview}
\label{subsec:protocol-overview}

\noindent\textbf{Protocol model.}
\sys replicates a deterministic state machine across $N=2f+1$ replicas and
tolerates up to $f$ crash failures.
Replicas follow the protocol but may crash or become unreachable. 
Byzantine behavior is outside of our scope.
The protocol guarantees linearizable execution
and preserves correctness under partial synchrony. 
When favorable strong synchrony exists,
\sys leverages it to achieve improved performance.

\noindent\textbf{Views and slot assignment.}
The protocol proceeds through monotonically numbered \emph{views}.
In each view, every replica is assigned a subset of the log slots and may
propose commands only for its assigned slots.
\sys assigns slots in a round-robin manner, making
each replica the leader for a disjoint subset of the log.
A view change reassigns the log space when the proposer set changes.

\noindent\textbf{Replica state.}
\autoref{fig:replica_state} summarizes the local state maintained by each replica.
Note that \var{last-append} is determined by the first empty slot in the log,
which is updated implicitly whenever a log entry is added or removed.
A replica acknowledges slot $s$ only after receiving \emph{all} proposals 
through $s$, i.e., the acknowledgment in \sys is \emph{cumulative}.

\noindent\textbf{Dual-mode transmission.}
The key challenge is to exploit strong synchrony for zero-bubble pipelining
without making correctness depend on this stronger timing property.
\sys thus separates replication semantics from transmission scheduling.
It uses \emph{pulsing mode} to realize the propagation-decoupled pipeline under
strong synchrony and falls back to \emph{responsive mode} under
partial synchrony.
Both modes operate on the same protocol state and safety rules, so switching
modes does not require protocol reconfiguration.

\subsection{Pulsing Mode}
\label{subsec:pulsing-mode}

Pulsing mode is the primary execution mode of \sys.
It maps the two-stage SMR pipeline onto propagation-decoupled rounds, which
process messages from earlier rounds while current messages propagate,
eliminating the \typeonebubble from the critical path.
This mode operates over the strong-synchrony network provided by \sys runtime (see \S\ref{sec:runtime}), 
where message delays are bounded and relatively stable.
Replicas can therefore advance through short, time-driven rounds without
waiting excessively for delayed messages, minimizing the \typetwobubble.
Together, propagation-decoupled rounds and the strong-synchrony runtime enable
zero-bubble pipelining in the common case.

In pulsing mode, replicas advance through coordinated rounds, 
i.e., each active replica multicasts at most one message per round, called a \emph{pulse}.
Each round's publication and consumption phases advance the proposal and acknowledgement stages of different log slots.

\autoref{fig:round_comparison} contrasts these round models with three replicas.
On the left, each replica processes peer messages published in the current
round, so the round must accommodate the slowest current-round propagation
path, producing the dashed propagation bubbles.
On the right, a replica processes messages published in different earlier
rounds according to when they become available.
The envelope labels indicate their publication rounds: messages processed
together need not have been transmitted in the same round.
A slower link therefore increases the round lag of messages arriving over that
link but does not extend every round.
Meanwhile, messages from the current round propagate for future processing
while replicas perform proposal or acknowledgment processing on messages that
have already arrived.

We detail the protocol operations of each phase in pulsing mode as follows.

\textbf{Publication phase.}
A client sends its request to any proposer; the default assignment lets every replica propose.
The replica buffers the request in \var{cmds} until its next proposal opportunity.
Replica $r$ constructs a pulse containing a proposal for \var{next-propose} and its cumulative acknowledgment \var{last-ack}.
After sending the pulse, $r$ advances \var{next-propose} to its next assigned
slot.
A pulse carries information from different pipeline instances.
It introduces a new proposal while reporting cumulative acknowledgment
progress for proposals processed in earlier rounds.
This aggregation removes the need for a separate acknowledgment transmission
for every proposal in the common case.

\textbf{Consumption phase.}
When a replica consumes a valid proposal from replica $r$, it records the command at the proposed slot and monotonically updates \var{acked}$[r]$ with the attached cumulative acknowledgment.
If the proposal fills the first hole in the local log, the replica advances
\var{last-append} and \var{last-ack} across the newly completed prefix.
A proposal received beyond an earlier hole may be buffered but cannot advance
the acknowledged prefix.
Let $q=f+1$ be the quorum size.
The $q$-th largest entry in \var{acked} identifies the longest log prefix
acknowledged by a quorum.
The replica advances \var{last-commit} to the smaller of this quorum index and its own acknowledged prefix.
It executes newly committed commands in order and replies to the clients whose requests it originally received.
The replica includes the updated \var{last-ack} in a subsequent pulse publication phase, allowing the acknowledgment stage for these proposals to overlap with the proposal stage for newer slots.
After the pipeline is filled, every round disseminates new proposals while
processing, acknowledging, and potentially committing proposals from earlier
rounds.
When all $R$ participating replicas are proposers, a normal round uses $R$ multicast messages to advance up to $R$ proposals.
The protocol therefore retains $O(1)$ amortized bottleneck message complexity
while overlapping network propagation with useful
protocol processing.

\subsection{Responsive Mode}
\label{subsec:responsive-mode}

When strong synchrony cannot be sustained, \sys falls back to responsive mode.
Responsive mode applies the same slot assignment, log, acknowledgement, and
commitment rules as pulsing mode, but sends protocol information immediately
rather than waiting for the next pulse.

When a client request arrives, its receiving replica immediately proposes the
request in its next assigned slot.
A replica receiving the proposal advances its acknowledged prefix when
possible and sends an explicit \textsc{Propose-Ack}.
Negative acknowledgements and recovery messages are likewise transmitted as
soon as they are generated.
Responsive mode sacrifices the message aggregation of pulsing mode but
preserves progress under partial synchrony.

Because the two modes differ only in transmission scheduling, switching modes
does not require a view change or log reconfiguration.
An acknowledgement piggybacked on a pulse and an explicit
\textsc{Propose-Ack} provide the same replication evidence and update the same
protocol state.

\subsection{Failure Recovery}
\label{subsec:failure-recovery}

Failures appear as stalled advancement of the committed prefix.
\sys first attempts lightweight recovery within the current view and performs
a view change when it suspects that a responsible proposer is unavailable.

\para{Detecting stalled progress.}
Each proposer maintains a coordination timer for its locally proposed but
uncommitted commands.
When the timer expires, the replica selects its earliest uncommitted proposal at slot $s$ and checks the
interval from $\textit{last-ack}+1$ through $s$.
If $\textit{last-ack}<s$, its local log is missing earlier proposals and it
requests those slots.
Otherwise, other replicas may have missed its proposal at $s$, so it
retransmits the proposal.
Responsive mode sends the resulting messages immediately, whereas pulsing
mode defers them to the next pulse.
The proposer resets the timer after each timeout action and whenever a local proposal commits.

\para{In-view recovery.}
A replica missing slot $s$ multicasts a \textsc{Propose-Nack} identifying the
missing slot and its view.
A replica retaining the proposal in its acknowledged prefix can respond with a
\textsc{Propose-Recover} carrying the proposal, allowing one multicast to
repair the same log hole at multiple replicas.
In pulsing mode, a proposer replacing a fresh proposal with recovery accounts for its next assigned slot with a no-op and advances its proposal pointer.
In responsive mode, an idle proposer fills its unused assigned slots on request; recovering an existing proposal does not itself consume a fresh slot.
A \textsc{Skip} message represents a range of no-ops compactly.
These mechanisms restore progress without changing the current view.

\para{Replica failure.}
If a proposer crashes or becomes partitioned, its assigned slots may remain
permanently empty and cannot be repaired within the current view.
A candidate then collects votes from a quorum using an at-least-as-up-to-date
log comparison.
Once elected as the new initiator, it proposes a \textsc{View-Init} entry that
excludes suspected failed replicas and defines a new slot assignment.
Followers reconcile their logs with the initiator before
\textsc{View-Init} commits, after which replicas install the new assignment
and resume normal operation.
The full view-change procedure and its safety argument appear in \autoref{sec:protocol:detail}.

\section{\runtime Runtime}
\label{sec:runtime}

\begin{figure}[t]
\centering
\includegraphics[width=\linewidth]{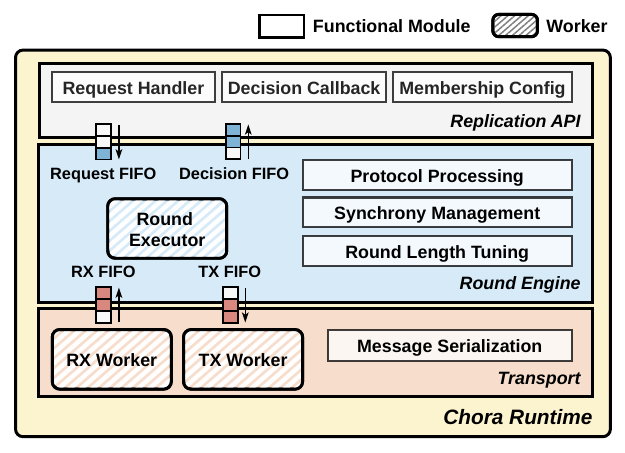}
\vspace{-24pt}
\caption{\runtime runtime architecture.}
\label{fig:system_arch}
\end{figure}

The \runtime runtime implements the propagation-decoupled rounds of \S\ref{sec:protocol} in physical time. It shrinks the \typetwobubble, which is the residual timing variation in processing and message arrival. As a result, each round lasts only as long as its expected work requires. Successive rounds then pack back-to-back. This yields zero-bubble common-case execution.
Strong synchrony provides the timing predictability this requires.
To this end, the runtime isolates timing-sensitive protocol operations on a
critical path (\S\ref{subsec:isolation}), adapts round boundaries to observed
timing (\S\ref{subsec:boundaries}), and falls back to responsive execution
when short coordinated rounds become ineffective (\S\ref{subsec:mode_switch}).
As shown in \autoref{fig:system_arch}, it consists of a replication API, a
round engine, and a transport layer connected by FIFO queues; one replica's
round engine also acts as the \emph{controller}, scheduling execution across
replicas---a role the view initiator (\S\ref{sec:protocol}) can host.


\subsection{Isolating Round Operations}
\label{subsec:isolation}

\noindent\textbf{Round critical path.}
Realizing short physical rounds requires the work that determines round
progress to have predictable cost.
The runtime therefore separates operations that determine round progress from 
operations that can proceed asynchronously.  
The critical path contains only the protocol operations required to 
complete the publication and consumption phases in \S\ref{subsec:pulsing-mode}.  
For publication, these operations select buffered requests, 
construct the next proposal, attach the cumulative acknowledgement,
and hand off the resulting pulse.  
For consumption, they apply available peer pulses to the log and 
acknowledgement state and advance the committed prefix.
In contrast, client admission, application execution, and packet I/O do not
directly advance the round and therefore remain outside its critical path.

\noindent\textbf{Critical-path isolation.}
The round engine executes the critical-path operations in a dedicated
\emph{round executor}, isolating them from application, transport, and
execution-resource interference.
First, the replication API performs client admission and application delivery 
outside the executor. 
\emph{Request} and \emph{decision FIFOs} connect the two, allowing the 
executor to consume admitted requests and return committed decisions without 
synchronizing with the application.
Second, dedicated \emph{RX} and \emph{TX workers} perform packet I/O and 
exchange messages with the executor through separate FIFOs. 
The executor hence updates protocol state without blocking on network I/O. 
Finally, the executor can run on a dedicated core, while the transport workers 
can use kernel-bypass I/O to reduce resource contention and packet-processing
variance~\cite{dpdk,arrakis,ix,demikernel}. 
Together, these isolation boundaries make the cost of 
critical-path operations more predictable.

\noindent\textbf{Arrival decoupling.}
Predictable processing alone is insufficient if each round must wait for a
particular network arrival.
The RX FIFO thus also decouples round execution from individual packet timing.
In the propagation-decoupled pipeline, the consumption phase may process pulses 
published in different earlier rounds. 
That is, it does not need to wait for all pulses from one designated round.  
The FIFO can absorb short arrival bursts and gaps while the executor consumes 
pulses that are already available.
Consequently, stable propagation delay changes the number of rounds for which
a pulse remains in flight rather than the round length.
Only arrival variation that drains the FIFO is exposed to the critical path.

\FloatBarrier
\subsection{Adaptive Round Boundaries}
\label{subsec:boundaries}

\begin{figure}[t]
\centering
\includegraphics[width=.75\linewidth]{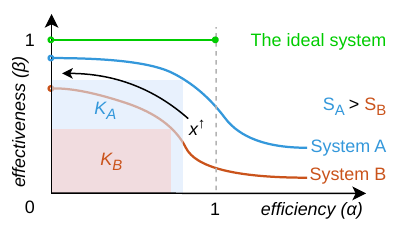}
\vspace{-16pt}
\caption{Pipeline efficiency--effectiveness trade-off
($S$ is the synchrony coefficient).}
\label{fig:efficiency_effectiveness}
\end{figure}

\para{Flexible round boundaries.}
Strong synchrony makes communication and processing delays tightly
concentrated, but not identical across rounds.
A fixed round must either cover rare tail delays or risk expiring during a 
protocol handler. 
We therefore use coordinated rounds with flexible termination
within a bounded timing window.
Specifically, \sys runtime sets a \emph{hard minimum round length} $\tau$
and a \emph{soft maximum deadline} $T$.  
The minimum limits the pulse publication rate.  
The maximum prevents the executor from waiting indefinitely for peer
pulses, but is soft because the executor checks it only between atomic protocol
operations.

Let the receive quota be one pulse from each of the other $N-1$
replicas. A round closes in one of four cases:
\begin{itemize}[leftmargin=*]
  \item The quota is complete before $\tau$: 
  the executor waits until $\tau$ and then closes the round.
  \item The quota becomes complete between $\tau$ and $T$: 
  the executor closes the round immediately.
  \item The quota is still incomplete at $T$ while the executor is polling:
  the executor closes the round at the deadline.
  \item A protocol operation finishes after $T$: 
  the executor closes the round immediately after that operation completes.
\end{itemize}
These rules allow common-case rounds to finish as soon as the expected protocol
operations are complete, while containing tail events without requiring every 
round to budget for them.  
The controller adapts both limits, keeping $\tau$ large enough to respect the configured pulse-rate limit.

\para{Optimization objective.}
Choosing the round limits trades pipeline \emph{efficiency} against \emph{effectiveness}.
Longer rounds can complete more useful work before advancing, but may spend
more time waiting; shorter rounds pack boundaries more tightly, but may defer
work to later rounds.
We formalize this trade-off 
as follows.  
For a continuously backlogged stream of $W$ commands, let an ideal pipeline commit the commands in $r_0$ rounds of length $t_0$.
For round limits $c=(\tau,T)$ with $\tau\leq T$, let $r(c)$ be the required number of rounds and $\bar t(c)$ their mean actual length.
We define
\[
\alpha(c)=\frac{t_0}{\bar t(c)},
\qquad
\beta(c)=\frac{r_0}{r(c)},
\]
where $\alpha$ is pipeline efficiency (how tightly rounds are packed in time)
and $\beta$ is pipeline effectiveness (useful progress per round).  
Their product gives throughput normalized to the ideal rate $\mathrm{Tput}_0$:
\[
\frac{\mathrm{Tput}(c)}{\mathrm{Tput}_0}
=\frac{W/[r(c)\bar t(c)]}{W/(r_0t_0)}
=\alpha(c)\beta(c).
\]
The peak normalized throughput over configurations $c$ is $\kappa=\max_c\alpha(c)\beta(c)$.
As shown in \autoref{fig:efficiency_effectiveness}, a more predictable
execution sustains higher effectiveness at a given efficiency and therefore
achieves a larger $\kappa$.
This connects strong synchrony directly to pipeline performance:
lower timing variation allows shorter rounds without sacrificing useful progress.

\para{Estimating the objective.}
The controller evaluates each candidate configuration against one fixed ideal reference.
For a given protocol, workload, replica count, and hardware, offline profiling
measures the active processing time of a fully productive round with all required inputs already available.
Its mean provides the reference round length $t_0$, while the protocol's ideal progress determines $r_0$ for the same amount of work.
For each tested configuration $c$, the runtime then measures two quantities:
the mean actual round length $\bar t(c)$, including waiting and overruns,
and the achieved committed progress per round, which determines $r(c)$.
Plugging these measurements into the definitions above yields estimated efficiency $\widehat\alpha(c)$ and effectiveness $\widehat\beta(c)$.
Their product estimates normalized throughput, and the highest value among tested configurations gives $\widehat\kappa = \max_c\widehat\alpha(c)\widehat\beta(c)$.


\para{Applying an update.}
The controller explores nearby configurations through bounded changes to
$\tau$ and $T$ and distributes each candidate to all replicas.
It retains a candidate only when its estimated improvement persists across
measurement windows;
otherwise, it restores the previous limits.
When samples or load are insufficient to evaluate a
candidate reliably, it leaves the current configuration unchanged.

\subsection{Synchrony-Aware Mode Switching}
\label{subsec:mode_switch}

\sys runtime uses pulsing mode while strong synchrony sustains short, coordinated
rounds and falls back to responsive mode when these timing conditions degrade.
To determine when to switch, each replica reports clock health
(synchronization status and error bound) together with pipeline health
(persistent overruns, quota misses, and messages per committed operation).
The controller evaluates these signals over observation windows.  
Sparse control pulses maintain measurements during low traffic.

The controller manages mode transitions across replicas without modifying
protocol state.  
When the monitored signals indicate sustained degradation, 
it disseminates a decision to switch to responsive mode, carrying a monotonic sequence number.  
On receipt, a replica stops publishing timed pulses and switches to 
responsive processing for messages.  
No view change or state conversion is needed because pulsing
and responsive execution use the same log and acknowledgment state and differ
only in transmission scheduling (\S\ref{subsec:responsive-mode}).  
After a longer healthy window, the controller schedules the group's return to
pulsing with a conservative initial deadline.  
Adaptation resumes once the round pipeline stabilizes.  
Different entry and return thresholds prevent repeated switching 
near the synchrony boundary.

\section{Generalizing the Pulsing Mode}
\label{sec:generalization}

Pulsing is not specific to \sys.
It applies to any replication system whose protocol progress can be summarized cumulatively
and published independently of when that progress is consumed.
We adapt such a system in two steps: identify the \emph{cumulative progress} whose latest value summarizes work across multiple operations,
then combine it with new work in propagation-decoupled pulses, so replicas can publish later pulses without waiting for the current one to be consumed.
This reorganizes how protocol work is communicated and scheduled while preserving the system's original replication semantics and completion rules.
We demonstrate the generalization on Derecho (\autoref{sec:generalization:derecho}) and Hermes (\autoref{sec:generalization:hermes}), two systems with different replication models.

\begin{figure*}[t]
\centering
\includegraphics[width=\textwidth]{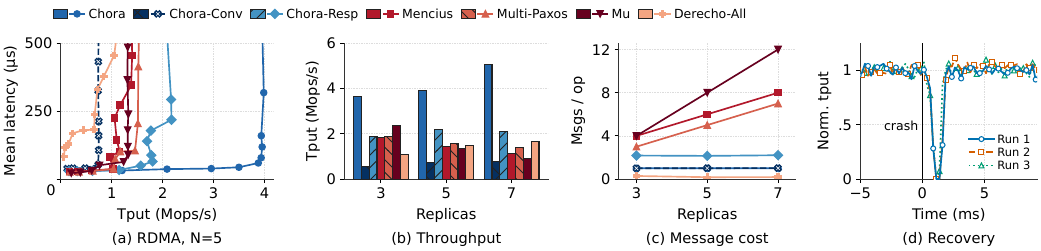}
\vspace{-23pt}
\caption{\sys sustains higher throughput as replicas increase and recovers after a crash.
Panel (c) counts TX messages plus RX processing steps triggered by peer updates per committed command at the busiest replica.
Each multicast counts once; Mu and Derecho count RDMA writes per destination.
Panel (d) normalizes each run to its mean throughput before the crash.}
\label{fig:end-to-end}
\end{figure*}


\subsection{Derecho}
\label{sec:generalization:derecho}

Derecho~\cite{derecho, spindle} provides ordered multicast by separating data dissemination from agreement on delivery.
Its protocol already maintains cumulative progress,
so only the second step is needed:
publishing this existing progress in propagation-decoupled pulses
without changing how Derecho orders or delivers messages.

\para{Cumulative progress.}
Derecho maintains shared protocol state in a \define{shared state table} (SST).
In its small-message path, a sender publishes the availability of new payloads,
while replicas advance cumulative SST counters as those payloads are received and processed.
Because later counter values subsume earlier ones, the latest SST state summarizes progress across many messages
and serves directly as Derecho's cumulative progress.

\para{Propagation-decoupled pulsing.}
Each pulse publishes the latest payload availability together with the sender's current cumulative progress.
Payloads are made available before the corresponding progress is published,
preserving Derecho's existing ordering between data and metadata.
After sending a pulse, a replica continues processing already available work
and may publish later pulses without waiting for responses to the current one.


To keep rounds short and predictable, we execute Derecho's native SST processing in local rounds on a dedicated thread,
isolated from communication work.
Protocol work need not align one-to-one with pulses:
Derecho retains its native batching and delivery mechanisms,
while each pulse publishes the latest cumulative progress.

\subsection{Hermes}
\label{sec:generalization:hermes}

Hermes~\cite{hermes} provides linearizable key-value replication using invalidations and per-key logical timestamps.
Unlike Derecho, Hermes tracks progress through acknowledgments and validations associated with individual writes,
so adaptation needs the first step as well:
constructing a cumulative representation of that progress.

\para{Cumulative progress.}
%
A Hermes write begins when a replica broadcasts an invalidation containing the key, timestamp, and new value;
the writer collects acknowledgments from every other replica before broadcasting a validation.
Logical timestamps resolve concurrent writes to the same key,
while early value propagation allows another replica to replay an interrupted write.
\cmt{Make sure these terms are from the original hermes paper.}
To make this per-write progress cumulative, we assign sequence numbers to invalidations within each origin worker's stream.
For every origin, a replica maintains a \emph{processed prefix}:
the largest sequence number up to which all invalidations from that origin have passed through the replica's native invalidation handler.
Because the prefix advances monotonically, one value summarizes processing of all preceding invalidations from that origin;
for each replica, these prefixes form a vector with one entry per origin.
In the same way, each origin maintains a \emph{completed prefix} covering writes for which it has collected all acknowledgments.

\para{Propagation-decoupled pulsing.}
%
%
Each pulse carries newly generated invalidations with their sequence numbers,
the replica's vector of processed prefixes, and the writer's completed prefix.
For each origin, the minimum processed prefix reported by all replicas identifies the prefix of writes that every replica has processed.
This cumulative evidence replaces the exchange of individual acknowledgments for the covered writes,
while the writer still runs Hermes's native acknowledgment logic, including its key, timestamp, and replay checks, on each newly covered write.
The writer's completed prefix plays the same role for validation:
it identifies the writes ready for validation without individual validation messages,
while a follower runs Hermes's native validation handler, including the per-key timestamp check, on covered writes up to its own processed prefix.

As with Derecho, each Hermes worker processes its protocol work in local rounds:
at the start of a round, it seals its prefixes and newly generated work and broadcasts them as a pulse,
then consumes available invalidations and progress until the round ends.

\FloatBarrier
\section{Evaluation}
\label{sec:evaluation}

We evaluate whether \sys reduces pipeline bubbles while retaining the low amortized bottleneck message complexity of round-based pipelining.
We compare end-to-end results (\autoref{sec:evaluation:performance}), isolate the mechanisms (\autoref{sec:evaluation:mechanisms}), and test its generality (\autoref{sec:evaluation:generality}).
Except when studying round length, we tune \sys's round limits following the objective in \autoref{subsec:boundaries}.

\begin{figure*}[t]
\centering
\includegraphics[width=\textwidth]{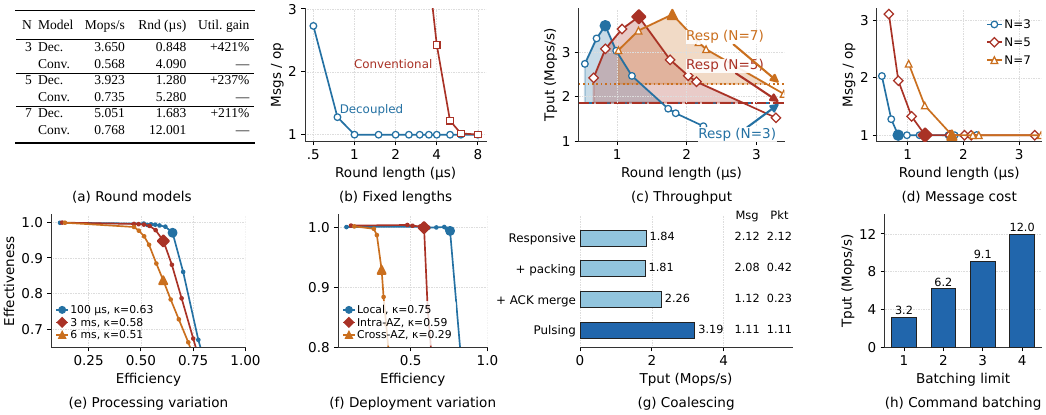}
\vspace{-20pt}
\caption{Reducing pipeline bubbles improves efficiency while retaining effectiveness; command batching adds further gains.
In (a), Dec./Conv. denote decoupled/conventional rounds; Util. gain is the relative utilization increase over conventional rounds.
Horizontal lines in (c) show responsive throughput; large markers in (e,f) maximize $\kappa$.
Message and packet metrics sum what all replicas send per committed command, counting each multicast once.}
\label{fig:mechanisms}
\end{figure*}

\subsection{Implementation and Experimental Setup}
\label{sec:evaluation:setup}


\para{Implementation.}
We implement \sys as a C library and use DPDK 25.11.0 for its runtime (23.11 on AWS).
The runtime supports RDMA (RoCEv2) and UDP.
We run RDMA with unreliable datagrams
(UD) for hardware multicast.

\para{Testbeds.}
The local cluster uses Xeon Gold 6226R/6326 CPUs (2.90\,GHz) and ConnectX-5 NICs.
We use four isolated cores per replica.
Cloud experiments use AWS EC2 m7i.2xlarge instances with spread placement and UDP, either within one availability zone (intra-AZ) or across zones (cross-AZ).

\para{Baselines.}
Our baselines isolate protocol choices from the effects of pulsing.
Multi-Paxos and Mencius provide single-leader and rotating multi-leader baselines.
\sys-Resp and \sys-Conv run our protocol with responsive execution and conventional rounds, respectively.
All four use \sys runtime.
\sys-Conv 
accepts only current-round pulses.
We also use the optimized open-source RDMA implementations of Mu~\cite{mu} and Derecho~\cite{derecho,spindle}, enabling all Derecho senders (Derecho-All).

\para{Workload.}
SMR clients spread evenly across proposers send 64-byte commands over kernel UDP on a separate network fabric.
Proposals carry at most one command except in Derecho, which retains native batching.
\autoref{sec:evaluation:mechanisms} separately studies command batching in \sys with its round policy fixed.

\subsection{End-to-End Performance}
\label{sec:evaluation:performance}

On local RDMA, \sys combines higher throughput with low-load latency close to responsive execution, sustains its advantage as replicas increase, and recovers after a crash.

\para{Latency and throughput with RDMA.}
With five replicas, \sys leads at 3.92\,Mops/s: 1.79$\times$ \sys-Resp and 2.87$\times$ Mu (\autoref{fig:end-to-end}b).
At low load, \sys and \sys-Resp respond in 29.4 and 29.2\,$\mu$s, respectively (\autoref{fig:end-to-end}a).
Short round crossings add little waiting, while propagation remains part of commit latency (\autoref{sec:case:analysis}).

\para{Scalability.}
\sys sustains its throughput advantage as the replica count grows from three to seven (\autoref{fig:end-to-end}b--c).
With seven replicas, it exceeds the strongest baseline, \sys-Resp, by $2.40\times$.
Round-based aggregation keeps amortized bottleneck message complexity near one message per committed command; Multi-Paxos and Mu handle more messages as replicas increase.

The \typeonebubble still limits \sys-Conv despite its low message cost, leaving processing capacity idle.
Derecho's native batching reduces message cost below one; \autoref{sec:evaluation:generality} evaluates coordinating its processing passes with pulsing.

\noindent\textbf{Crash recovery.}
After a replica crash, \sys resumes processing in about 2\,ms and returns near its previous throughput (\autoref{fig:end-to-end}d).
We stop one replica's DPDK application near peak throughput in each of three runs.
Heartbeat timers trigger a view change to exclude the failed replica; the survivors then shorten their round timeouts.

\FloatBarrier
\subsection{Pipelining Mechanisms and Timing}
\label{sec:evaluation:mechanisms}

\begin{figure*}[t]
\centering
\includegraphics[width=\textwidth]{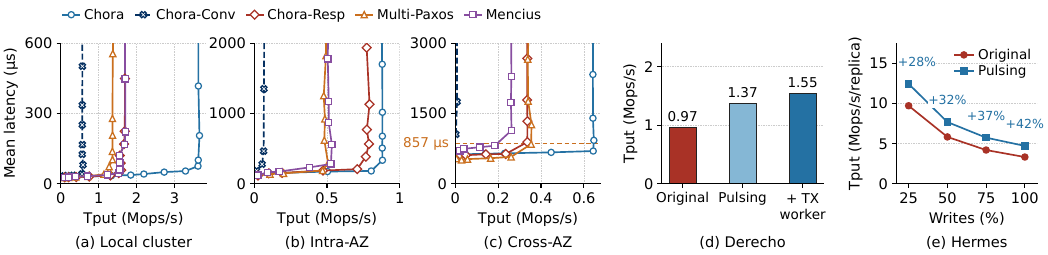}
\vspace{-18pt}
\caption{Pulsing improves throughput across deployments (a--c) and systems (d,e) with three replicas.
Panels (d,e) use fixed membership and best tested originals at saturation.
The dashed line in (c) marks the lowest measured mean latency for Multi-Paxos clients farthest from the leader.}
\label{fig:generality}
\end{figure*}

We test how decoupling and synchrony reduce pipeline bubbles to improve efficiency while retaining effectiveness, then compare pulsing with coalescing and batching.

\noindent\textbf{Propagation-decoupled rounds.}
Propagation-decoupled rounds remove the \typeonebubble by overlapping propagation with processing of pulses from earlier rounds.
Both round models share the protocol, runtime, and throughput-feedback controller, with fixed membership and offered load and no batching or coalescing.
Only decoupled rounds use early exit.
Reducing time lost to pipeline bubbles leaves more CPU time for protocol processing.
Across three to seven replicas, the CPU utilization rate for protocol processing increases by 211\%--421\% over conventional rounds (\autoref{fig:mechanisms}a).
Throughput improves by $5.3$--$6.6\times$.

Without tuning or early exit, decoupled rounds retain low message cost at shorter lengths (\autoref{fig:mechanisms}b).
Decoupled rounds send about one pulse per operation from 1\,$\mu$s upward; conventional rounds need about 6\,$\mu$s to approach that cost.
Shorter conventional rounds discard late pulses and retransmit proposals, making no progress at 2.5\,$\mu$s and below.
Decoupled rounds also incur extra pulses when too short.

\noindent\textbf{Round length.}
Throughput peaks when rounds are just long enough for processing, balancing pipeline efficiency and effectiveness (\autoref{fig:mechanisms}c--d).
Short rounds leave operations unfinished, lowering effectiveness and increasing messages per commit.  
Long rounds leave capacity idle, lowering efficiency.
Pulsing exceeds responsive throughput throughout the shaded ranges, so gains extend beyond one deadline setting.
\sys uses this trade-off to select deadlines (\autoref{subsec:boundaries}).

\noindent\textbf{Synchrony.}
%
Stronger synchrony sustains higher pipeline effectiveness at a given efficiency by reducing the \typetwobubble (\autoref{fig:mechanisms}e).
We inject uniformly sampled receive-thread delays with ranges of 100\,$\mu$s, 3\,ms, and 6\,ms in a local three-replica deployment.
Deployment measurements show the same trade-off. 
Peak normalized throughput falls from 0.75 locally to 0.59 within an AZ and
0.29 across AZs (\autoref{fig:mechanisms}f).

\noindent\textbf{Relation to coalescing and batching.}
Pulsing adds gains beyond coalescing and complements command batching.
With three RDMA replicas, each panel uses 64-byte proposals for coalescing and
256-byte proposals for batching.
Coalescing reduces traffic but does not match pulsing's throughput (\autoref{fig:mechanisms}g).
Packing reduces packets; cumulative ACK merging also removes superseded acknowledgments without changing proposal processing.
Both use a 10\,$\mu$s interval, selected for ACK merging from a 0.5--10\,$\mu$s sweep.
ACK merging raises responsive throughput from 1.84 to 2.26\,Mops/s; pulsing reaches 3.19\,Mops/s, 41\% higher despite similar logical traffic.
Traffic totals sum across senders.
Command batching raises throughput without changing the round policy (\autoref{fig:mechanisms}h).
With up to four commands per proposal, throughput rises from 3.19 to 11.97\,Mops/s.
The executor takes queued commands without waiting to fill batches; sparse load still produces single-command proposals.

\FloatBarrier
\subsection{Generality}
\label{sec:evaluation:generality}

Pulsing also improves throughput across deployments and when applied to Derecho and Hermes (\autoref{sec:generalization}).

\para{Cloud deployments.}
Pulsing retains its throughput advantage in the cloud despite longer and more variable network delays.
We compare local, intra-AZ, and cross-AZ deployments over UDP, excluding the RDMA-based Mu and Derecho baselines (\autoref{fig:generality}a--c).
We keep all systems within cloud packet-rate allowances using longer \sys rounds and smaller responsive windows where needed.
Within one AZ, \sys reaches 0.88\,Mops/s versus 0.79 for \sys-Resp.
Across AZs, it reaches 0.65\,Mops/s versus 0.36 for Multi-Paxos, the strongest baseline there.
\sys-Conv reaches only 0.07\,Mops/s within an AZ and 0.009\,Mops/s across AZs, because its lock-step rounds must cover the cloud's long propagation delays as well as protocol processing.
Accepting commands at every replica lets clients use a nearby replica in wider deployments without routing through a fixed leader.
\sys's cross-AZ low-load latency is about 620\,$\mu$s across zones; Multi-Paxos ranges from 402\,$\mu$s in its leader's zone to 857\,$\mu$s in the farthest zone (dashed line).

\para{Derecho Adaption.}
Pulsing improves Derecho's throughput by coordinating its existing processing passes.
Throughput rises from 0.966 to 1.367\,Mops/s, a 41.6\% gain without extra workers (\autoref{fig:generality}d).
The metric counts globally ordered deliveries across senders.
A TX worker adds 13.4\%, reaching 1.550\,Mops/s.
We retain native batching and compare against the best tested original from a configuration sweep.
Pulsing thus complements Derecho's native batching.

\para{Hermes Adaption.}
Pulsing raises Hermes throughput by 28.1--41.8\%, with greater gains as writes increase (\autoref{fig:generality}e).
This trend is consistent with sharing acknowledgment and validation work across writes.
Throughput counts completed requests per replica, including local reads.
Both variants use two workers per replica without communication offload, accessing one million keys uniformly with 31-byte values.
We give both implementations equal tuning effort and compare their best tested configurations at each write fraction.

\section{Discussion and Related Work}
\label{sec:related}


\para{Network models.}
Most SMR protocols assume partial synchrony~\cite{partialsync1988} for liveness~\cite{flp,paxosms2001,vr2012,raft2014,pbft}.
The timed asynchronous model~\cite{timedasync} supports reasoning about timing without full synchrony.
Speculative Paxos~\cite{specpaxos2015} provides mostly-ordered multicast through SDN.
Network sequencers reduce coordination in NOPaxos~\cite{nopaxos2016}, Eris~\cite{eris}, and Hydra~\cite{hydra}.

%
\para{Datacenter replication.}
Datacenter replication systems optimize communication and data movement, and recent work also exploits datacenter synchrony.
DARE~\cite{dare}, APUS~\cite{apus}, and Mu~\cite{mu} optimize replication and consensus with RDMA.
Hermes~\cite{hermes}, Derecho~\cite{derecho}, and Spindle~\cite{spindle} exploit concurrency, pipelining, and batching.
The Synchronous Data Center~\cite{syncdc2019} argues for fully synchronous data centers, and Watchmaker~\cite{watchmaker2026} and SSR~\cite{ssr} build SMR on synchronized time slots or rounds.
These systems primarily optimize communication and data movement;
\sys instead exploits predictable communication and processing timing
to coordinate SMR execution more efficiently.

\para{Hardware acceleration.}
\sys can benefit from hardware acceleration, including embedding time synchronization in protocol messages via NIC hardware timestamps, reducing processing tail latency through careful message handling, and offloading consensus to programmable hardware such as SDN switches~\cite{netpaxos} or FPGAs~\cite{consensusbox2016,ssr}.

\section{Conclusion}


We present \sys, a message-efficient SMR protocol co-designed with a network runtime to realize zero-bubble pipelining. First, its propagation-decoupled
pipeline processes earlier-round messages while publishing current proposals,
eliminating the propagation bubble without sacrificing message efficiency.
Second, the runtime isolates timing-sensitive operations and adapts round
boundaries, minimizing the jitter bubble under strong synchrony. When strong
synchrony degrades, \sys falls back to responsive execution while preserving
correctness and progress under ordinary synchrony.

\label{sec:conclusion}

\balance
\bibliographystyle{ACM-Reference-Format}
\bibliography{references}

@inproceedings{tempo2021,
  title     = {Efficient replication via timestamp stability},
  author    = {Enes, Vitor and Baquero, Carlos and Gotsman, Alexey and Sutra, Pierre},
  year      = {2021},
  month     = {Apr},
  booktitle = {Proceedings of the Sixteenth European Conference on Computer Systems},
  publisher = {ACM},
  series    = {EuroSys ’21},
  pages     = {178–193},
  doi       = {10.1145/3447786.3456236},
  url       = {http://dx.doi.org/10.1145/3447786.3456236},
  collection = {EuroSys ’21}
}

@article{paxosms2001,
  title   = {Paxos Made Simple},
  author  = {Lamport, Leslie},
  journal = {ACM SIGACT News},
  volume  = {32},
  number  = {4},
  pages   = {51--58},
  month   = dec,
  year    = {2001}
}

@techreport{vr2012,
  title       = {Viewstamped {Replication} {Revisited}},
  author      = {Liskov, Barbara and Cowling, James},
  number      = {MIT-CSAIL-TR-2012-021},
  institution = {MIT},
  month       = jul,
  year        = {2012}
}

@inproceedings{raft2014,
  title     = {In {Search} of an {Understandable} {Consensus} {Algorithm}},
  isbn      = {978-1-931971-10-2},
  booktitle = {Proceedings of the {USENIX} {Conference} on {USENIX} {Annual} {Technical} {Conference}},
  publisher = {USENIX Association},
  author    = {Ongaro, Diego and Ousterhout, John},
  year      = {2014},
  series    = {{USENIX} {ATC} '14},
  pages     = {305--319}
}

@article{pmmc2015,
  title     = {Paxos Made Moderately Complex},
  author    = {Van Renesse, Robbert and Altinbuken, Deniz},
  journal   = {ACM Computing Surveys},
  volume    = {47},
  number    = {3},
  pages     = {1--36},
  month     = feb,
  year      = {2015},
  publisher = {Association for Computing Machinery},
  issn      = {0360-0300},
  doi       = {10.1145/2673577}
}

@article{linearizability1990,
  title     = {Linearizability: A correctness condition for concurrent objects},
  author    = {Herlihy, Maurice P and Wing, Jeannette M},
  journal   = {ACM Transactions on Programming Languages and Systems (TOPLAS)},
  volume    = {12},
  number    = {3},
  pages     = {463-492},
  year      = {1990},
  publisher = {ACM New York, NY, USA}
}

@inproceedings{specpaxos2015,
  title     = {Designing distributed systems using approximate synchrony in data center networks},
  author    = {Ports, Dan RK and Li, Jialin and Liu, Vincent and Sharma, Naveen Kr and Krishnamurthy, Arvind},
  booktitle = {12th USENIX Symposium on Networked Systems Design and Implementation (NSDI 15)},
  pages     = {43-57},
  year      = {2015}
}

@article{spanner2013,
  title     = {Spanner: Google’s globally distributed database},
  author    = {Corbett, James C and Dean, Jeffrey and Epstein, Michael and Fikes, Andrew and Frost, Christopher and Furman, Jeffrey John and Ghemawat, Sanjay and Gubarev, Andrey and Heiser, Christopher and Hochschild, Peter and others},
  journal   = {ACM Transactions on Computer Systems (TOCS)},
  volume    = {31},
  number    = {3},
  pages     = {1-22},
  year      = {2013},
  publisher = {ACM New York, NY, USA}
}

@inproceedings{chubby2006,
  title     = {The Chubby lock service for loosely-coupled distributed systems},
  author    = {Burrows, Mike},
  booktitle = {Proceedings of the 7th symposium on Operating systems design and implementation},
  pages     = {335-350},
  year      = {2006}
}

@inproceedings{hotstuff2019,
  title     = {HotStuff: BFT consensus with linearity and responsiveness},
  author    = {Yin, Maofan and Malkhi, Dahlia and Reiter, Michael K and Gueta, Guy Golan and Abraham, Ittai},
  booktitle = {Proceedings of the 2019 ACM Symposium on Principles of Distributed Computing},
  pages     = {347-356},
  year      = {2019}
}

@article{partialsync1988,
  title     = {Consensus in the presence of partial synchrony},
  author    = {Dwork, Cynthia and Lynch, Nancy and Stockmeyer, Larry},
  journal   = {Journal of the ACM (JACM)},
  volume    = {35},
  number    = {2},
  pages     = {288-323},
  year      = {1988},
  publisher = {ACM New York, NY, USA}
}

@inproceedings{nopaxos2016,
  title     = {Just Say {NO} to {Paxos} Overhead: Replacing Consensus with Network Ordering},
  author    = {Li, Jialin and Michael, Ellis and Sharma, Naveen Kr and Szekeres, Adriana and Ports, Dan RK},
  booktitle = {12th USENIX Symposium on Operating Systems Design and Implementation (OSDI 16)},
  pages     = {467-483},
  year      = {2016}
}

@inproceedings{mencius2008,
  title     = {Mencius: building efficient replicated state machines for WANs},
  author    = {Mao, Yanhua and Junqueira, Flavio P. and Marzullo, Keith},
  booktitle = {Proceedings of the 8th USENIX Conference on Operating Systems Design and Implementation},
  pages     = {369–384},
  year      = {2008},
  publisher = {USENIX Association},
  series    = {OSDI'08}
}

@inproceedings{zookeeper2010,
  title     = {{ZooKeeper}: Wait-free Coordination for Internet-scale Systems},
  author    = {Hunt, Patrick and Konar, Mahadev and Junqueira, Flavio P and Reed, Benjamin},
  booktitle = {2010 USENIX Annual Technical Conference (USENIX ATC 10)},
  year      = {2010}
}

@inproceedings{pbft,
  series    = {{OSDI} '99},
  title     = {Practical {Byzantine} {Fault} {Tolerance}},
  booktitle = {Proceedings of the {Third} {Symposium} on {Operating} {Systems} {Design} and {Implementation}},
  publisher = {USENIX Association, Co-sponsored by IEEE TCOS and ACM SIGOPS},
  author    = {Castro, Miguel and Liskov, Barbara},
  month     = feb,
  year      = {1999}
}

@inproceedings{narwhal-tusk,
  author    = {Danezis, George and Kokoris-Kogias, Lefteris and Sonnino, Alberto and Spiegelman, Alexander},
  title     = {Narwhal and Tusk: a DAG-based mempool and efficient BFT consensus},
  year      = {2022},
  isbn      = {9781450391627},
  publisher = {Association for Computing Machinery},
  doi       = {10.1145/3492321.3519594},
  booktitle = {Proceedings of the Seventeenth European Conference on Computer Systems},
  pages     = {34–50},
  numpages  = {17},
  series    = {EuroSys '22}
}

@inproceedings{arrakis,
  series    = {{OSDI} '14},
  title     = {Arrakis: {The} {Operating} {System} is the {Control} {Plane}},
  isbn      = {978-1-931971-16-4},
  booktitle = {Proceedings of the 11th {USENIX} {Conference} on {Operating} {Systems} {Design} and {Implementation}},
  publisher = {USENIX Association},
  author    = {Peter, Simon and Li, Jialin and Zhang, Irene and Ports, Dan R. K. and Woos, Doug and Krishnamurthy, Arvind and Anderson, Thomas and Roscoe, Timothy},
  year      = {2014},
  pages     = {1-16}
}

@inproceedings{ix,
  series    = {{OSDI} '14},
  title     = {{IX}: {A} {Protected} {Dataplane} {Operating} {System} for {High} {Throughput} and {Low} {Latency}},
  isbn      = {978-1-931971-16-4},
  booktitle = {Proceedings of the 11th {USENIX} {Conference} on {Operating} {Systems} {Design} and {Implementation}},
  publisher = {USENIX Association},
  author    = {Belay, Adam and Prekas, George and Klimovic, Ana and Grossman, Samuel and Kozyrakis, Christos and Bugnion, Edouard},
  year      = {2014},
  pages     = {49-65}
}

@inproceedings{demikernel,
  series    = {{SOSP} '21},
  title     = {The {Demikernel} {Datapath} {OS} {Architecture} for {Microsecond}-{Scale} {Datacenter} {Systems}},
  isbn      = {978-1-4503-8709-5},
  doi       = {10.1145/3477132.3483569},
  booktitle = {Proceedings of the {ACM} {SIGOPS} 28th {Symposium} on {Operating} {Systems} {Principles}},
  publisher = {Association for Computing Machinery},
  author    = {Zhang, Irene and Raybuck, Amanda and Patel, Pratyush and Olynyk, Kirk and Nelson, Jacob and Leija, Omar S. Navarro and Martinez, Ashlie and Liu, Jing and Simpson, Anna Kornfeld and Jayakar, Sujay and Penna, Pedro Henrique and Demoulin, Max and Choudhury, Piali and Badam, Anirudh},
  year      = {2021},
  pages     = {195-211}
}

@misc{dpdk,
  author       = {Linux Foundation},
  title        = {{Data Plane Development Kit}},
  year         = {2024},
  howpublished = {\url{https://www.dpdk.org/}}
}

@article{flp,
  series  = {{JACM} '85},
  title   = {Impossibility of {Distributed} {Consensus} with {One} {Faulty} {Process}},
  volume  = {32},
  issn    = {0004-5411},
  doi     = {10.1145/3149.214121},
  number  = {2},
  journal = {Journal of the ACM},
  author  = {Fischer, Michael J. and Lynch, Nancy A. and Paterson, Michael S.},
  month   = apr,
  year    = {1985},
  pages   = {374-382}
}

@article{smr,
  author     = {Schneider, Fred B.},
  title      = {Implementing fault-tolerant services using the state machine approach: a tutorial},
  year       = {1990},
  issue_date = {Dec. 1990},
  publisher  = {Association for Computing Machinery},
  volume     = {22},
  number     = {4},
  issn       = {0360-0300},
  doi        = {10.1145/98163.98167},
  journal = {ACM Computing Surveys},
  month      = {dec},
  pages      = {299–319}
}

@inproceedings{gfs,
  series    = {{SOSP} ’03},
  title     = {The {Google} {File} {System}},
  isbn      = {1-58113-757-5},
  doi       = {10.1145/945445.945450},
  booktitle = {Proceedings of the {Nineteenth} {ACM} {Symposium} on {Operating} {Systems} {Principles}},
  publisher = {Association for Computing Machinery},
  author    = {Ghemawat, Sanjay and Gobioff, Howard and Leung, Shun-Tak},
  year      = {2003},
  pages     = {29-43}
}

@inproceedings{eris,
  series    = {{SOSP} '17},
  title     = {Eris: {Coordination}-{Free} {Consistent} {Transactions} {Using} {In}-{Network} {Concurrency} {Control}},
  isbn      = {978-1-4503-5085-3},
  doi       = {10.1145/3132747.3132751},
  booktitle = {Proceedings of the 26th {Symposium} on {Operating} {Systems} {Principles}},
  publisher = {Association for Computing Machinery},
  author    = {Li, Jialin and Michael, Ellis and Ports, Dan R. K.},
  year      = {2017},
  pages     = {104-120}
}

@inproceedings{consensusbox2016,
  title     = {Consensus in a box: Inexpensive coordination in hardware},
  author    = {Istv{\'a}n, Zsolt and Sidler, David and Alonso, Gustavo and Vukolic, Marko},
  booktitle = {13th USENIX Symposium on Networked Systems Design and Implementation (NSDI 16)},
  pages     = {425-438},
  year      = {2016}
}

@inproceedings{epaxos2021,
  title     = {{EPaxos} Revisited},
  author    = {Tollman, Sarah and Park, Seo Jin and Ousterhout, John},
  booktitle = {18th USENIX Symposium on Networked Systems Design and Implementation (NSDI 21)},
  pages     = {613-632},
  year      = {2021}
}

@inproceedings{clockrsm2014,
  title        = {Clock-RSM: Low-latency inter-datacenter state machine replication using loosely synchronized physical clocks},
  author       = {Du, Jiaqing and Sciascia, Daniele and Elnikety, Sameh and Zwaenepoel, Willy and Pedone, Fernando},
  booktitle    = {2014 44th Annual IEEE/IFIP International Conference on Dependable Systems and Networks},
  pages        = {343-354},
  year         = {2014},
  organization = {IEEE}
}

@inproceedings{gaios,
  author    = {Bolosky, William J. and Bradshaw, Dexter and Haagens, Randolph B. and Kusters, Norbert P. and Li, Peng},
  title     = {Paxos replicated state machines as the basis of a high-performance data store},
  year      = {2011},
  publisher = {USENIX Association},
  booktitle = {Proceedings of the 8th USENIX Conference on Networked Systems Design and Implementation},
  pages     = {141–154},
  numpages  = {14},
  series    = {NSDI'11}
}

@inproceedings{hydra,
  author    = {Inho Choi and Ellis Michael and Yunfan Li and Dan R. K. Ports and Jialin Li},
  title     = {Hydra: {Serialization-Free} Network Ordering for Strongly Consistent Distributed Applications},
  booktitle = {20th USENIX Symposium on Networked Systems Design and Implementation (NSDI 23)},
  year      = {2023},
  isbn      = {978-1-939133-33-5},
  pages     = {293-320},
  url       = {https://www.usenix.org/conference/nsdi23/presentation/choi},
  publisher = {USENIX Association},
  month     = apr
}

@inproceedings{epaxos,
  author    = {Moraru, Iulian and Andersen, David G. and Kaminsky, Michael},
  title     = {There is more consensus in Egalitarian parliaments},
  year      = {2013},
  isbn      = {9781450323888},
  publisher = {Association for Computing Machinery},
  doi       = {10.1145/2517349.2517350},
  booktitle = {Proceedings of the Twenty-Fourth ACM Symposium on Operating Systems Principles},
  pages     = {358–372},
  numpages  = {15},
  series    = {SOSP '13}
}

@inproceedings{netpaxos,
  title     = {{NetPaxos}: Consensus at Network Speed},
  author    = {Dang, Huynh Tu and Sciascia, Daniele and Canini, Marco and Pedone, Fernando and Soul{\'e}, Robert},
  booktitle = {Proceedings of the 1st ACM SIGCOMM Symposium on Software Defined Networking Research},
  pages     = {1--7},
  year      = {2015}
}

@inproceedings{smr-scalability,
  title={State machine replication scalability made simple},
  author={Stathakopoulou, Chrysoula and Pavlovic, Matej and Vukoli{\'c}, Marko},
  booktitle={Proceedings of the Seventeenth European Conference on Computer Systems},
  pages={17--33},
  year={2022}
}

@book{lynch1996distributed,
  title={Distributed algorithms},
  author={Lynch, Nancy A},
  year={1996},
  publisher={Elsevier}
}

@article{derecho,
  title     = {Derecho: Fast State Machine Replication for Cloud Services},
  author    = {Jha, Sagar and Behrens, Jonathan and Gkountouvas, Theo and Milano, Mae and Song, Weijia and Tremel, Edward and Van Renesse, Robbert and Zink, Sydney and Birman, Ken},
  journal   = {ACM Transactions on Computer Systems (TOCS)},
  volume    = {36},
  number    = {2},
  pages     = {1--49},
  year      = {2019},
  publisher = {ACM New York, NY, USA}
}

@inproceedings{spindle,
  title     = {Spindle: Techniques for Optimizing Atomic Multicast on {RDMA}},
  author    = {Jha, Sagar and Rosa, Lorenzo and Birman, Ken},
  booktitle = {2022 IEEE 42nd International Conference on Distributed Computing Systems (ICDCS)},
  pages     = {1085--1097},
  month     = jul,
  year      = {2022},
  publisher = {IEEE},
  doi       = {10.1109/ICDCS54860.2022.00108}
}

@inproceedings{mu,
  title     = {Microsecond Consensus for Microsecond Applications},
  author    = {Aguilera, Marcos K. and Ben-David, Naama and Guerraoui, Rachid and Marathe, Virendra J. and Xygkis, Athanasios and Zablotchi, Igor},
  booktitle = {14th USENIX Symposium on Operating Systems Design and Implementation (OSDI 20)},
  pages     = {599--616},
  year      = {2020}
}

@inproceedings{hermes,
  title     = {Hermes: A Fast, Fault-Tolerant and Linearizable Replication Protocol},
  author    = {Katsarakis, Antonios and Gavrielatos, Vasileios and Katebzadeh, M.R. Siavash and Joshi, Arpit and Dragojevic, Aleksandar and Grot, Boris and Nagarajan, Vijay},
  booktitle = {Proceedings of the Twenty-Fifth International Conference on Architectural Support for Programming Languages and Operating Systems (ASPLOS)},
  pages     = {201--217},
  year      = {2020},
  publisher = {ACM}
}

@inproceedings{dare,
  title     = {DARE: High-Performance State Machine Replication on RDMA Networks},
  author    = {Poke, Marius and Hoefler, Torsten},
  booktitle = {Proceedings of the 24th International Symposium on High-Performance Parallel and Distributed Computing (HPDC)},
  pages     = {107--118},
  year      = {2015},
  publisher = {ACM}
}

@inproceedings{apus,
  title     = {APUS: Fast and Scalable Paxos on RDMA},
  author    = {Wang, Cheng and Jiang, Jianyu and Chen, Xusheng and Yi, Ning and Cui, Heming},
  booktitle = {Proceedings of the 2017 Symposium on Cloud Computing (SoCC)},
  pages     = {94--107},
  year      = {2017},
  publisher = {ACM}
}

@article{timedasync,
  title     = {The Timed Asynchronous Distributed System Model},
  author    = {Cristian, Flaviu and Fetzer, Christof},
  journal   = {IEEE Transactions on Parallel and Distributed Systems},
  volume    = {10},
  number    = {6},
  pages     = {642--657},
  year      = {1999},
  publisher = {IEEE}
}

@inproceedings{barnoy1991,
  title     = {Optimal amortized distributed consensus},
  author    = {Bar-Noy, Amotz and Deng, Xiaotie and Garay, Juan A. and Kameda, Tiko},
  booktitle = {Proceedings of the 5th International Workshop on Distributed Algorithms (WDAG)},
  pages     = {95--107},
  year      = {1991},
  publisher = {Springer}
}

@inproceedings{dagrider,
  title     = {All You Need is DAG},
  author    = {Keidar, Idit and Kokoris-Kogias, Lefteris and Naor, Oded and Spiegelman, Alexander},
  booktitle = {Proceedings of the 2021 ACM Symposium on Principles of Distributed Computing},
  pages     = {165--175},
  year      = {2021}
}

@inproceedings{bullshark,
  title     = {Bullshark: DAG BFT Protocols Made Practical},
  author    = {Spiegelman, Alexander and Giridharan, Neil and Sonnino, Alberto and Kokoris-Kogias, Lefteris},
  booktitle = {Proceedings of the 2022 ACM SIGSAC Conference on Computer and Communications Security},
  pages     = {2705--2718},
  year      = {2022}
}

@article{nezha,
  title     = {Nezha: Deployable and High-Performance Consensus Using Synchronized Clocks},
  author    = {Geng, Jinkun and Sivaraman, Anirudh and Prabhakar, Balaji and Rosenblum, Mendel},
  journal   = {Proceedings of the VLDB Endowment},
  volume    = {16},
  number    = {4},
  pages     = {629--642},
  month     = dec,
  year      = {2022},
  publisher = {VLDB Endowment},
  issn      = {2150-8097},
  doi       = {10.14778/3574245.3574250}
}

@inproceedings{ssr,
  title={SSR: Synchronous State Machine Replication with a Timing-Aligned Split Architecture},
  author={Ma, Yuke and Costa, Paolo and Xia, Yiting},
  booktitle={Proceedings of the 17th ACM SIGOPS Asia-Pacific Workshop on Systems},
  pages={67--73},
  year={2026},
  doi={10.1145/3838177.3841728}
}

@inproceedings{syncdc2019,
  title     = {The Synchronous Data Center},
  author    = {Yang, Tian and Gifford, Robert and Haeberlen, Andreas and Phan, Linh Thi Xuan},
  booktitle = {Proceedings of the Workshop on Hot Topics in Operating Systems},
  series    = {HotOS '19},
  pages     = {142--148},
  month     = may,
  year      = {2019},
  publisher = {ACM},
  doi       = {10.1145/3317550.3321442}
}

@inproceedings{watchmaker2026,
  title     = {Running Distributed Systems like Clockwork},
  author    = {Newatia, Karan and Gifford, Robert and Lu, Qingjie and Haeberlen, Andreas and Phan, Linh Thi Xuan},
  booktitle = {1st New Ideas in Networked Systems (NINeS 2026)},
  series    = {OASIcs},
  volume    = {139},
  pages     = {26:1--26:31},
  year      = {2026},
  publisher = {Schloss Dagstuhl -- Leibniz-Zentrum f{\"u}r Informatik},
  doi       = {10.4230/OASIcs.NINeS.2026.26}
}
\clearpage
\nobalance
\appendix

\section{Modeling Synchrony for SMR Pipelines}
\label{sec:network}

Coordinating the SMR pipeline requires a timing model for replica communication and processing.
We define synchrony metrics that relate variation in these delays to the round length.
Using these metrics, we analyze the performance of conventional round-based protocol constructs.
We then introduce a new network round model to further improve the processing efficiency.
Next, we discuss trade-offs of the \sys round model to better understand its performance.
Lastly, we describe a new network-level API based on the new round model.
A new replication protocol that builds on this model and API is introduced in \autoref{sec:protocol}.

\subsection{Quantifying System Synchrony}
\label{sec:network:metric}

In \autoref{sec:introduction}, we introduced the synchrony coefficient $S \in (0,1]$ as a simplified metric for network synchrony.
We now give a formal definition that accounts for the percentile of the delay distribution.
Similar to prior models, our synchrony definition centers around message delays in a networked system.
However, we further refine the model to define a \emph{degree} of synchrony, instead of a fixed delay bound.
Suppose the message delay is represented by a random variable $d$.
We define a \emph{$x^{th}$ synchrony coefficient} as:
\begin{equation}\label{eq:coefficient_old}
    \tilde{S}^{x}  = \frac{\mu (d)}{p_{x^{th}}(d)}
\end{equation}
where $p_{x^{th}}(Z)$ and $\mu (Z)$ denote the $x^{th}$ percentile and the expected value of random variable $Z$, respectively.

This coefficient quantifies the synchrony degree.
$p_{x^{th}}(d) = \frac{1}{\tilde{S}^{x}} \mu (d)$ specifies a time bound for receivers to complete processing a message after transmission, with an expectation of $x\%$.
$\frac{1}{\tilde{S}^x} - 1$ represents the required relative time buffer to tolerate the tail delays.
If the network is perfectly synchronous, $\mu (d) = p_{x^{th}}(d)$, which results in $\tilde{S}^x = 1$.

Nodes in prior fully synchronous protocols move in lock-step rounds.
This enables synchronous and coordinated behavior across the network.
In such a round, each node multicasts a message in the beginning, and then receives and processes the messages from other nodes.
We define $\Delta \tilde{T}^{x}$ to be the time bound that a node can finish all processing for a round with an expectation of $x\%$.
In the normal case, a node can receive and process more than one message in each round.
Therefore, $\Delta p_{x^{th}}(d)$ is the lower bound of $\Delta \tilde{T}_{x}$.
Assuming that the system can complete a workload of $\Delta W_{x}$ on average in each round with such an expectation, then the system throughput upper bound can be described as:
\begin{equation}\label{eq:tput_old}
   \textit{Tput}^{x} = \frac{\Delta W^{x}}{\Delta \tilde{T}^{x}} \leq \tilde{S}^{x} \frac{\Delta W^{x}}{\mu (d)}
\end{equation}
From the equation, we can see that a more synchronous system (with a larger $S_x$) provides better performance.

\subsection{Propagation-Decoupled Rounds}
\label{sec:network:round}

Let's look at $\Delta \tilde{T}_{x}$ from another perspective.
We further decompose the message delay $d$ into a network induced propagation delay $d_{prop}$ and a delay introduced by the application as a processing delay $d_{proc}$.
Note that $d_{prop}$ includes queuing delay, transmission delay, processing delay in the network stack, and propagation delay in the physical medium.
With $d_{prop}$ and $d_{proc}$, we have:

\begin{equation}\label{eq:round_length_old}
\begin{split}
    \Delta \tilde{T}_{x} \geq p_{x^{th}}(d) \approx {} & p_{x^{th}}(d_{prop}) + p_{x^{th}}(d_{proc}) \\
                 = {} & \mu (d_{prop}) + \left[p_{x^{th}}(d_{prop}) - \mu (d_{prop})\right] + \\
                      & \mu (d_{proc}) + \left[p_{x^{th}}(d_{proc}) - \mu (d_{proc})\right]
\end{split}
\end{equation}

The equation indicates that $\Delta \tilde{T}_x$ is bounded by both the expected and the tail of both propagation delay and processing delay.
With this observation, we introduce \emph{\sys rounds} to decrease its time bound ($\Delta T_x$) for better performance.

Similar to a conventional synchronous round, in a \sys round, a node also first multicasts a message and then receives and processes messages from other nodes.
However, a \sys round doesn't require nodes to process messages from the same round.
This relaxation removes the expectation of propagation delay from the lower bound of $\Delta T_x$.
In other words, for a \sys round, we have:
\begin{equation}\label{eq:rount_length}
\begin{split}
    \Delta T_{x} \gtrsim & \left[p_{x^{th}}(d_{prop}) - \mu (d_{prop})\right] + \\
                         & \mu (d_{proc}) + \left[p_{x^{th}}(d_{proc}) - \mu (d_{proc})\right]
\end{split}
\end{equation}

\autoref{fig:network} shows an example of a system operating on \sys rounds.
Each message is attached with a number that denotes the round when it is sent.
In the next round ($i+1$), node B will process node A's message from round $i$, while node C will process node A's message from round $i-3$.
Our new round structure allows the round length to be \emph{decoupled} from the longest propagation delay (node A to C).

In practice, processing delays are often much lower than propagation delays on the critical path.
For instance, on our local cluster one-way propagation averages 2.32\,$\mu$s, while a replica spends 0.77\,$\mu$s on average processing a round's messages.
With decoupled rounds, five replicas run rounds of about 1\,$\mu$s and reach 5.3$\times$ the throughput of conventional rounds.

With the new round definition and \autoref{eq:rount_length}, we rectify \autoref{eq:coefficient_old} as:
\begin{equation}\label{eq:coefficient}
    S^x = \frac{\mu(d_{proc})}{P_{x^{th}}(d_{prop}) - \mu(d_{prop}) + P_{x^{th}}(d_{proc})}
\end{equation}

Similar to \autoref{eq:tput_old}, we have the following formula that describes the relationship between $S^x$ and the system's performance:
\begin{equation}\label{eq:tput}
   \textit{Tput}^{x} = \frac{\Delta W^{x}}{\Delta T^{x}} \leq S^{x} \frac{\Delta W^{x}}{\mu (d_{proc})}
\end{equation}
The equation shows that a more synchronous system, which has a higher $S^x$, is able to provide better performance.
For a perfectly synchronous system where tail delays equal average delays, we have $P_{x^{th}}(d_{prop}) = \mu(d_{prop})$ and $P_{x^{th}}(d_{proc}) = \mu(d_{proc})$.
In this case, the throughput upper bound is $\frac{\Delta W^x}{\mu(d_{proc})}$.
This upper bound denotes a maximum processing utility rate of 100\%.

\subsection{Pipeline Performance}
\label{sec:network:analysis}

The throughput bound in \autoref{eq:tput} trades work completed per round against round length.
\autoref{subsec:boundaries} develops this trade-off as pipeline efficiency and effectiveness and uses their product to motivate runtime deadline tuning.
For the strict periodic rounds modeled here, actual and configured round lengths are equal.
The runtime formulation also accounts for early exit and processing that crosses a soft deadline.

\section{Detailed \sys Protocol}
\label{sec:protocol:detail}

This appendix specifies the messages and state transitions underlying the protocol in \autoref*{sec:protocol}.
\autoref*{fig:replica_state} in the main paper lists the replica state.
It covers normal operation, recovery, and view changes; \autoref{sec:protocol:detail:safety} gives the safety argument.
The message descriptions give logical fields and actions, independent of how a transport packs them into packets.

\subsection{State, Views, and Slot Ownership}
\label{sec:protocol:detail:state}

\sys replicates a log across $N=2f+1$ replicas, with quorum size $q=f+1$.
Replicas follow the protocol and may crash; progress requires an eventual period of ordinary synchrony and a reachable quorum.
Strong synchrony enables pulsing but is not a safety assumption.
Each log entry contains a command, its slot, and the view in which its proposer created it.
A command may contain one or more client requests or a protocol operation such as \cmd{view-init} or \cmd{no-op}.
Slot~0 contains an initial null command with view~0.

\para{Prefix indices.}
All indices below are inclusive: \var{last-append} is the slot immediately before the first hole in the log, and \var{last-commit} is the last locally committed slot.
The index \var{last-ack} marks the prefix the replica can acknowledge in its current view.
During normal operation, it follows \var{last-append}; during a view change, a replica must first reconcile its uncommitted prefix with the new initiator.
Thus, \var{last-commit}$\leq$\var{last-ack}$\leq$\var{last-append}.
The array \var{acked} stores the largest acknowledgment received from each replica in the current view, including the local replica's own acknowledgment.

\para{Slot ownership.}
A committed \cmd{view-init} establishes the proposer set and slot assignment for its view.
The slot \var{view-base} contains this command; ordinary proposals start at \var{view-base}$+1$.
For $R$ proposers with offsets $i\in\{1,\ldots,R\}$, the round-robin assignment gives proposer $i$ slots \var{view-base}$+i+nR$ for nonnegative integers $n$.
A proposer creates commands only for its assigned slots and never creates two different commands for the same slot in the same view.
Its \var{next-propose} identifies its first unused assigned slot; advancing this pointer selects the next slot in its assignment.
A replica with no assigned slots still receives proposals and contributes acknowledgment evidence.

\para{Message admission.}
Every replica message identifies its sender and current view.
A replica ignores messages from older views and applies the transition rules in \autoref{sec:protocol:detail:viewchange} before processing messages that establish a newer view.
Normal processing requires an established slot assignment; retransmissions preserve the original slot, command, and command view.
Duplicate messages do not create a second log entry or decrease an acknowledgment index.
A replica never replaces an existing command with a different command in the same view.

\subsection{Messages}
\label{sec:protocol:detail:messages}

The following fields describe the logical protocol messages.
A recovery message carries the original proposal's view as part of its command, separately from the message's current view.
An implementation may combine several logical records in one pulse.

\begin{description}[leftmargin=0pt,labelsep=.5em]
\item[\msg{Request} and \msg{Reply}]
A client sends $\langle\msg{Request},req\text{-}id,op\rangle$ to a proposer in the current view.
With the default assignment, every replica is a proposer.
The identifier distinguishes requests for at-most-once execution; retries retain the same identifier.
The receiving replica buffers the request in \var{cmds} and returns $\langle\msg{Reply},req\text{-}id,result\rangle$ after executing it.
The execution layer retains completed request identifiers and results, and answers a retry with the prior result instead of executing the operation again.
\item[\msg{Propose}]
The fields are $view$, $log\text{-}slot$, $ack\text{-}slot$, and $cmds$.
They identify the proposal's view and slot, the sender's cumulative acknowledgment, and the command payload.
\item[\msg{Propose-Ack}]
The fields identify the current view and the sender's cumulative acknowledgment index $ack\text{-}slot$.
The acknowledgment covers every slot through that index.
\item[\msg{Propose-Nack}]
The fields are $view$ and $nack\text{-}slot$; a sender may combine requests for multiple missing slots.
\item[\msg{Propose-Recover}]
The fields are $view$, $recover\text{-}slot$, and $cmd$, where $cmd$ retains the original proposal and its command view.
\item[\msg{Propose-Noop}]
The fields are $view$ and $noop\text{-}slot$, denoting a no-op in a slot owned by the sender.
\item[\msg{Skip}]
The fields are $view$, $slot\text{-}start$, and $slot\text{-}until$.
The sender proposes no-ops for its own assigned slots in that inclusive range.
\item[View-change messages]
A \msg{View-Change-Request} carries the new view and the candidate's last contiguous slot and its command view.
A \msg{View-Change-Vote} identifies the new view and chosen candidate.
The elected initiator proposes a \cmd{View-Init} command containing the new slot assignment, as specified in \autoref{sec:protocol:detail:viewchange}.
\end{description}

\subsection{Normal Operation in Pulsing Mode}
\label{sec:protocol:detail:pulsing}

Each pulsing round has a publication phase and a consumption phase.
A replica publishes at most one pulse per round; a pulse can carry a proposal, cumulative acknowledgment, and recovery or skipping records.
The publication and consumption phases advance the proposal and acknowledgment stages of different log slots.

\para{Publication phase.}
A proposer selects buffered requests, constructs a \msg{Propose} for \var{next-propose}, and includes its current \var{last-ack} as $ack\text{-}slot$.
It records its own proposal locally, multicasts the pulse, and advances \var{next-propose}.
A proposer with no client command can fill its assigned slot with a no-op, including through the skipping rule in \autoref{sec:protocol:detail:skipping}.
A replica without an assigned proposal slot publishes its acknowledgment and any pending recovery information without creating a new proposal.

\para{Consumption phase.}
The replica consumes available pulses from earlier rounds while current pulses propagate.
For a valid proposal from replica $r$, it records $cmds$ at $log\text{-}slot$ and updates \var{acked}$[r]$ to the larger of its current value and $ack\text{-}slot$.
A proposal beyond a hole remains buffered; it cannot advance the contiguous prefix.
When a proposal fills the first hole, the replica scans the completed prefix to update \var{last-append} and, in an established view, \var{last-ack}.
It also updates its own entry in \var{acked}.
The next publication carries the resulting acknowledgment progress.

\para{Commitment and execution.}
Let $c$ be the $q$-th largest index in \var{acked}.
A quorum has acknowledged every slot through $c$.
In an established view, the replica advances \var{last-commit} through the smaller of $c$ and its locally acknowledged prefix, never decreasing the committed index.
The local bound prevents execution of commands that the quorum knows but this replica has not yet received.
It executes newly committed commands in slot order, suppresses repeated execution of the same client request identifier, and returns results for requests it originally received.
Before the new view's \cmd{view-init} commits, the separate commitment barrier in \autoref{sec:protocol:detail:viewchange} applies.

\subsection{Responsive Mode}
\label{sec:protocol:detail:responsive}

Responsive mode uses the same log and quorum rules but publishes messages as their triggering events occur.
A proposer with a newly buffered request creates and multicasts a \msg{Propose}, records the proposal locally, and advances \var{next-propose}.
On receiving a proposal for slot $s$, a replica records it and advances its contiguous prefix when possible.
It multicasts \msg{Propose-Ack} only if \var{last-ack}$\geq s$; otherwise, the proposal remains behind an earlier hole.
An outgoing message carrying the same cumulative acknowledgment can replace a separate ACK transmission.
An ACK receiver monotonically updates the sender's entry in \var{acked} and applies the commitment rule.

\para{Unused slots.}
An idle proposer must not leave assigned slots empty while other proposers advance.
If a NACK requests its own unused slot $s\geq\var{next-propose}$, it proposes no-ops for its assigned slots from \var{next-propose} through $s$ and advances the pointer past them.
A \msg{Skip} can represent these no-ops compactly.
Recovering a previously proposed command does not itself allocate a fresh local slot in responsive mode.

\para{Mode transitions.}
A mode change changes when replicas transmit, while preserving the view, slot ownership, and meaning of acknowledgment evidence.
A replica can consume either an explicit ACK or acknowledgment information carried by a pulse using the same state-update rule.
The runtime coordinates transitions as described in \autoref*{subsec:mode_switch}.
Neither transition permits a replica to acknowledge a missing proposal or bypass the view-change rules.

\subsection{Proposal Recovery}
\label{sec:protocol:detail:recovery}

A missing proposal prevents a replica from acknowledging the log beyond the first hole.
A replica requests slot $s$ with \msg{Propose-Nack}; in pulsing mode it attaches the request to its next pulse, and in responsive mode it transmits immediately.
The NACK can accompany a new proposal and does not itself consume that proposal's slot.

\para{Recovery eligibility.}
A replica can answer with \msg{Propose-Recover} only if it retains the requested proposal in its acknowledged prefix, so $s\leq\var{last-ack}$.
This restriction prevents a follower from distributing an unverified old-view suffix during reconciliation.
The original proposer supplies missing current-view commands; other eligible replicas can also supply the same proposal.
One recovery multicast can repair the same hole at several replicas.
A receiver checks the view, records the recovered proposal, and advances the contiguous and acknowledged prefixes under the normal rule or the reconciliation rule, as appropriate.

\para{Pulsing recovery.}
The recovery record occupies the next pulse's command payload instead of a fresh client command.
A sender with an active assigned slot accounts for that slot with \msg{Propose-Noop} and advances \var{next-propose}; a \msg{Skip} can encode the same no-op together with other unused slots.
A sender without an assigned slot recovers the command without proposing a no-op.
Thus, recovery does not abandon a newly consumed proposal slot.
In responsive mode, the sender transmits the recovery immediately without this companion no-op.

\subsection{Coordination and Failure Timers}
\label{sec:protocol:detail:timers}

Each proposer starts a coordination timer when it installs its slot assignment.
The timer drives its locally proposed commands, including no-ops, toward commitment.
At a timeout, it selects the earliest locally proposed slot $s>\var{last-commit}$.
If no such slot exists, it resets the timer without transmitting a repair.
Otherwise, it checks the interval from \var{last-ack}$+1$ through $s$.
If the interval contains holes, it requests all missing proposals with NACKs.
If \var{last-ack}$\geq s$, it retransmits its original proposal for $s$ because another replica may have missed it.
Responsive mode transmits immediately; pulsing mode schedules the repair in a subsequent pulse.
The proposer resets the timer after handling a timeout and whenever a local proposal commits.

A separate failure timeout detects lack of progress from a proposer and triggers a view-change attempt on suspicion of failure.
Failure suspicion need not be accurate for safety.
Candidates retry unsuccessful elections in higher views, using randomized backoff to resolve concurrent attempts.
\autoref{sec:protocol:detail:viewchange} gives the state transitions.

\subsection{Skipping and Proposer Configuration}
\label{sec:protocol:detail:skipping}

\para{Eager skipping.}
Each replica tracks \var{latest-propose-slot}, the largest proposal slot it knows, and includes this progress hint in outgoing replica messages.
A proposer uses received hints to identify unused assigned slots that precede other replicas' progress.
It fills those slots with no-ops before its next transmission and represents a consecutive range of its assigned slots with \msg{Skip}.
It advances \var{next-propose} beyond the range and never skips a slot for which it already proposed a different command.
A receiver expands the range according to the sender's assignment and applies the normal prefix and commitment rules.
In responsive mode, it acknowledges a completed skipped range when its cumulative prefix reaches the last skipped slot.

\para{Proposer configuration.}
A view may assign slots to all replicas or only a subset.
Changing the proposer set requires a view change; it does not change the quorum size for the fixed replica group.
With one proposer, the remaining replicas still acknowledge its commands, so progress depends on a quorum rather than new proposals from every replica.
The $R$-pulse, $R$-proposal accounting for the normal multi-proposer case assumes that all participating replicas have proposal slots; acknowledgment-only participants must also publish their evidence when the proposer set is smaller.

\subsection{View Change}
\label{sec:protocol:detail:viewchange}

A replica starts a view change when it suspects that a proposer has failed or become unreachable.
The new view selects an initiator whose log preserves committed commands and installs a slot assignment that excludes suspected failed proposers.
Quorum intersection and log reconciliation preserve safety even when a suspicion is incorrect.

\para{Starting an election.}
A candidate $r'$ sets \var{view} to a higher number, \var{role} to candidate, and \var{voted-for} to itself.
It clears buffered proposals beyond \var{last-append}, sets \var{next-propose} to $null$, resets peer entries in \var{acked} to zero, and sets \var{last-ack} to \var{last-commit}.
It records its own acknowledged prefix in \var{acked}, clears \var{voted-by}, and inserts its own identifier.
These updates suspend ordinary proposals and prevent acknowledgment evidence from different views from mixing.
The candidate multicasts a \msg{View-Change-Request} containing the new view number, its \var{last-append}, and the command view at that slot.

\para{Voting.}
A replica receiving a request for a higher view compares the candidate's contiguous log with its own.
For logs ending at slots $s$ and $s'$ with command views $v$ and $v'$, respectively, the first log is at least as up-to-date as the second exactly when
\[
    v>v' \quad\text{or}\quad (v=v'\ \text{and}\ s\geq s').
\]
If the candidate's log fails this check, the receiver starts an election in a still higher view.
Otherwise, it adopts the requested view, becomes a follower, records the candidate in \var{voted-for}, and applies the same proposal-suspension and prefix-reset operations as the candidate.
It replies to that candidate with \msg{View-Change-Vote} and does not vote for another candidate in the same view.
The candidate counts distinct votes for itself in the current view in \var{voted-by}.
Once this set reaches $q$ replicas, it becomes the initiator.

\para{Proposing the new assignment.}
The initiator sets \var{view-base} to one more than the last contiguous slot in the log on which it won the election.
It proposes \cmd{View-Init} at \var{view-base}, retaining all earlier commands and their original command views.
This command defines the proposer set and slot assignment beginning at \var{view-base}$+1$.
The initiator acknowledges its complete prefix through \var{view-base} but leaves \var{next-propose} as $null$ until \cmd{View-Init} commits locally.
Thus, it cannot create ordinary proposals under an assignment that has not yet committed.

\para{Adopting the initiator.}
A replica learning the elected initiator's proposal in a higher view adopts that view and becomes its follower.
A candidate or follower learning the elected initiator in its current view adopts that initiator, even if it voted for a different candidate.
It suspends ordinary proposals, discards buffered entries beyond its contiguous prefix, and resets acknowledgment evidence to its committed prefix as above.
A replica already following that initiator in the same view does not repeat these resets for subsequent messages.
On first receiving this initiator's \cmd{View-Init}, the follower sets \var{view-base} to its slot and removes any uncommitted suffix starting there.
It retains the new command separately until it reconciles the preceding prefix.
Duplicate copies of the same \cmd{View-Init} do not repeat truncation or reset acknowledgment state, including after the command commits.
It retains its earlier contiguous commands but cannot acknowledge an uncommitted inherited suffix merely because its log contains those commands.

\para{Reconciling the prefix.}
The follower requests proposals beyond \var{last-ack} through \var{view-base}$-1$ using \msg{Propose-Nack}.
A recovery responder must have acknowledged the requested slot.
For a recovered proposal at slot $s$ with \var{last-ack}$<s\leq$\var{last-append}, the follower compares its retained command view with the recovered command view.
If they match, it advances \var{last-ack} through $s$; the log-matching property in \autoref{claim:log_matching} establishes agreement on that prefix.
If they differ, it removes its uncommitted suffix starting at $s$, reducing \var{last-append} to $s-1$, and requests the missing suffix again.
Once its acknowledged prefix reaches its appended prefix, it can append further recovered proposals and advance both indices.
After reconciling all slots below \var{view-base}, it appends the pending \cmd{View-Init} and acknowledges that slot.
The new command cannot itself serve as evidence that the inherited prefix matches the initiator.
Recovery preserves each proposal's original command view, even though the enclosing message carries the new view.
A follower publishes its updated acknowledgment so that the initiator and peers can observe reconciliation progress.

\para{Committing the view.}
Let $c$ be the $q$-th largest current-view acknowledgment index.
A replica uses this quorum evidence to advance commitment only when $c\geq\var{view-base}$: the quorum must acknowledge \cmd{View-Init} and its entire preceding prefix.
The replica must also have locally reconciled that prefix through \var{view-base}.
It then commits \cmd{View-Init} together with earlier uncommitted commands, executes newly committed client commands in order, and installs the new assignment.
Its \var{next-propose} becomes its first assigned ordinary slot, or remains $null$ if it has no assigned slots.
Subsequent commitment follows the normal quorum and local-prefix bounds.
This guard prevents new-view acknowledgments for an incomplete inherited prefix from committing commands before the replica establishes the new view.

\para{Retrying and rejecting stale messages.}
A candidate that cannot obtain a quorum before its election timeout retries in a higher view with randomized backoff.
Replicas ignore lower-view messages throughout election, reconciliation, and normal operation.
They continue recovery requests while reconciliation is incomplete, and can initiate another view change if the selected initiator stops making progress.

\subsection{Safety Proof}
\label{sec:protocol:detail:safety}

We consider the problem of replicating a \define{log} across $N=2f+1$ \define{replicas} using the \sys protocol specified in \autoref{sec:protocol:detail}.
Each replica has a local log.
The log consists of an infinite series of \define{slots}, starting from index 1.
Replicas propose \define{commands} for the log slots.
Each log slot can be either empty, or contain one unique command.
We present the formal proof of the following theorem:

\begin{theorem} \label{theorem:safety}
\emph{Safety:} If a command $c$ is committed at slot $s$, no other command can be committed at $s$.
\end{theorem}

Let's first clarify some conventions for ease of illustration.
When \emph{any log} is referred to, we are discussing all possible local logs of any replica at any given time.
For a certain log, we denote the command at slot $s$ using $s.cmd$ and. 
For a certain command $c$, we denote its view as $c.view$.
We say that a command at slot $s$ is an \emph{appended command} if $s \neq \var{append-slot}$.
When a command is \emph{appended} to a log at slot $s$, we denote the cases where the command is added to the log at $s$, and it becomes an appended command after the addition.  
When a command is \emph{committed}, it denotes that the command is locally committed in some log.
Also, we assume that there is a $null$ command at slot 0 of all logs, and the command's view is 0.

In \sys, there is a basic requirement for all replicas, which is trivial because no malicious node is considered:

\begin{fact}\label{fact:unique_proposal}
If a replica $r$ proposes a command $c$ at slot $s$ in view $v$, it never proposes a different command at $s$ in $v$.
\end{fact}

Also, from the protocol, we can notice that:
\begin{fact}\label{fact:initiator_add}
An initiator only adds commands from the current view to its log. 
\end{fact}

\begin{fact}\label{fact:add}
A command can and only can be added to a log in two cases:
\begin{enumerate}
  \item A replica learns a proposal from the same view with a \msg{propose} or a \msg{propose-recover}.
  \item A replica learns the command with a \msg{propose-recover} from an initiator.
\end{enumerate}
\end{fact}

\begin{claim}\label{claim:propose_then_add}
If command $c$ is added to any log in view $v$, then $c$ has been proposed and $c.view \leq v$.
\end{claim}

\begin{proof}
Without loss of generality, assume $v$ is the smallest view that $c$ is added.
Let's discuss the two possible cases in \autoref{fact:add}.

In the first case $c$ is proposed in the current view ($c.view = v$).

In the second case, given \autoref{fact:initiator_add}, $c$ is added before view $v$.
This violates the assumption that $v$ is the smallest view that $c$ is added to a log.
\end{proof}

\begin{claim}\label{claim:unique_initiator}
If a replica $r$ is elected as an initiator in view $v$, then no other replica can be elected as an initiator of $v$.
\end{claim}

\begin{proof}
If another replica $r'$ is elected as an initiator in $v$, both $r$ and $r'$ have received at least $f+1$ \msg{view-change-vote}s from different replicas.
Due to quorum intersection, there exists at least one replica $r''$ that has sent \msg{view-change-vote} to both $r$ and $r'$.
This violates the protocol where a follower ignores a \msg{view-change-request} from any other candidate in the same view after updating \var{voted-for} before replying \msg{view-change-vote}. 
\end{proof}

\begin{claim}\label{claim:unique_initview}
If a \cmd{init-view} for view $v$ exists in any log, then no other \cmd{init-view} for $v$ exists in any log.
\end{claim}

\begin{proof}
If the target \cmd{init-view} exists in a log, it must have been added to the log and thus must have been proposed (\autoref{claim:propose_then_add}).
In a view, only an initiator proposes a \cmd{init-view}.
Given \autoref{fact:unique_proposal} and \autoref{claim:unique_initiator}, the proposed \cmd{init-view} is unique.
\end{proof}

To simplify the following illustration, let's introduce more conventions at this point.
If a replica is elected in the view, we call it \emph{the initiator} of the view and denote it with $v.initiator$.
Correspondingly, we use \emph{the base log} of $v$ to denote the local log of the initiator when it starts election in $v$, if the initiator exists.
Also, we use $v.init$ to denote the unique \cmd{init-view} of $v$ if it exists.

\begin{claim}\label{claim:view_unique_proposal}
If command $c$ is proposed for slot $s$ in view $v$, then no other command can be proposed for $s$ in $v$.  
\end{claim}

\begin{proof}
Replicas propose in $v$ following $v.init$.
Given \autoref{claim:unique_initview} and \autoref{fact:unique_proposal}, we know that every proposal is unique.
\end{proof}

\begin{claim}\label{claim:initial_append}
If command $c$ is appended at slot $s$ following $c'$, then $c$ is appended at slot $s$ following $c'$ in $v = c.cmd$ for some log.
\end{claim}

\begin{proof}
Without loss of generality, let's discuss the smallest view $v'$ where this appending happens.
Given \autoref{claim:propose_then_add}, we have $v' \ge v$.
Let's discuss the two cases in \autoref{fact:add} separately.

For the first case, the view is $v$ ($v' = v$) when the appending happens.

For the second case, the checking logic of the follower before appending ensures that the append only happens if the command in the previous slot is consistent with the initiator's log.
If the claim doesn't hold, then we know $c$ is not appended after $c'$ to the initiator's log in any view $v'' < v'$.
Also, we know that $c$ is not appended to $c'$ in $v'$ given \autoref{fact:initiator_add}.
So, the initiator's log doesn't contain $c$ at $s$ following $c'$.
This leads to a contradiction.

With this serving as an induction step, we can know that if the targeted append doesn't happen in view $v$, it cannot happen in any view bigger than $v$.
\end{proof}

\begin{claim}\label{claim:monolithic_view}
For any log, if $s < s' \leq \var{last-append}$, then $s.cmd.view \leq s'.cmd.view$.
\end{claim}

\begin{proof}
From \autoref{claim:initial_append}, it is known that $s'.cmd$ is appended after $s.cmd$ in $s'.cmd.view$ for some log.
Assume that $s.cmd.view = v$ and $s'.cmd.view = v'$.
Assume that for this log, $s.cmd$ is appended in view $v''$.
From \autoref{claim:propose_then_add} we have $v \leq v''$.
The fact that $s.cmd$ exists when $s'.cmd$ is appended denotes that $v'' \leq v'$.
So, we have $v \leq v'$.
\end{proof}

\autoref{claim:monolithic_view} also implies the following statement:

\begin{claim}\label{claim:log_view_consecutive}
In a log, if $s < s' \leq \var{append-slot}$ and $s.cmd.view = s'.cmd.view = v$, then for any $s''$ that satisfies $s < s'' < s'$, $s''.cmd.view = v$.
In other words, the commands with the same view are consecutive in a log.    
\end{claim}

\begin{claim}\label{claim:initview_first}
For any log, if slot $s$ is the smallest appended slot in view $v$ that satisfies $s.cmd.view = v$, then $s.cmd$ is $v.init$.
\end{claim}

\begin{proof}
From the protocol, we can see that a replica only appends other commands from the current view after the current view's \cmd{view-init} has been appended.
Also, a replica never removes any command from the current view.
So, $s.cmd = v.init$.
\end{proof}

\begin{claim}\label{claim:slot_view_unique_append}
If an appended command $c$ exists at slot $s$ in some log, then for any log, if an appended command $c'$ exists at slot $s$ and $c.view = c'.view$, then $c = c'$.
Also, $c$ and $c'$ are appended following the same command.
\end{claim}

\begin{proof}
Assume $c.view = c'.view = v$.
According to \autoref{claim:initial_append}, both appendings happen in $v$.
According to \autoref{claim:view_unique_proposal}, $c = c'$. 

Assume $c$ is appended following $c_p$.
If $c_p.view = v$, then according to the analysis above, it is unique.
Otherwise, because of \autoref{claim:log_view_consecutive} and \autoref{claim:initview_first}, $c = v.init$.
According to the protocol, the check of followers when appending $v.init$ in $v$ ensures that $c_p$ is unique and is consistent with the initiator's log.
Given \autoref{claim:initial_append}, $c_p$ is unique.
\end{proof}

\begin{claim}\label{claim:log_matching}
\emph{Log Matching Property:} If two logs contain an identical appended command with the same slot $s$ and view $v$, then the two logs are identical up to slot $s$. 
\end{claim}

\begin{proof}
If the claim doesn't hold, without loss of generality, assume $s'$ is the biggest slot that violates the claim.
In other words, $s' \leq s$ and the commands are not identical.
Then the commands at $s' + 1$ are identical ($c$) and are appended.
When looking at $c$, \autoref{claim:slot_view_unique_append} is violated because it is appended to two different commands.
\end{proof}

\begin{claim}\label{claim:log_inclusion_lemma}
Assume the last appended slot of $l$ and $l'$ are $s$ and $s'$ respectively, if $s \ge s'$ and $s.view = s'.view = v$, then the appended commands of the two logs ($c$ and $c'$) at $s'$ are identical. 
\end{claim}

\begin{proof}
If $c.view = c'.view$, according to \autoref{claim:slot_view_unique_append}, we have $c = c'$.
Otherwise, given \autoref{claim:log_view_consecutive} and \autoref{claim:initview_first}, we know that for $l$, $v.init$ is in a slot bigger than $s'$, while for $l'$, $v.init$ is in a slot no bigger than $s'$.
This violates \autoref{claim:slot_view_unique_append}.
\end{proof}

\begin{claim}\label{claim:log_inclusion}
\emph{Log Inclusion Property:} Assume the last appended slot of $l$ and $l'$ are $s$ and $s'$ respectively, if $s \ge s'$ and $s.view = s'.view = v$, then the commands at a slot no bigger than $s'$ are identical in the two logs for any $s'' \leq s'$.
In other words, $l$ includes all appended commands in $l'$. 
\end{claim}

\begin{proof}
According to \autoref{claim:log_inclusion_lemma}, we know that the appended commands at $s'$ are identical.
Given \autoref{claim:log_matching}, we know that the two logs are identical up to slot $s'$. 
\end{proof}

\begin{claim}\label{claim:initiator_completeness_lemma}
If $c$ is committed in view $v$ and $c.view = v$, then $c$ is included in the log of the initiator of any view $v'$ that satisfies $v' > v$.
\end{claim}

\begin{proof}
If the claim doesn't hold, without loss of generality, assume $v'$ is the smallest view that violates the claim.
Since $c$ is committed in view $v$, so at least $f+1$ replicas append $c$ in $v$ (given \autoref{claim:propose_then_add}, $c$ can not be appended in a view smaller than $v$).
Also, at least $f+1$ replicas vote for $v'.initiator$ in $v'$.
So, there exists one replica $r$ that both appends $c$ in $v$ and votes for $v'.initiator$ in $v'$.
Also, according to the protocol, a replica only removes a command when it is inconsistent with the initiator's log,
and the leaders for views between $v$ and $v'$ all include $c$ in the log (since $v'$ is the smallest view that doesn't satisfy the claim).
So, $c$ is still included in the log of $r$ when it votes for $v'.initiator$.

Assume the last appended command in the log of $r$ is $c_1$ when it votes for $v'.initiator$.
Also, assume the last appended command in the base log of $v'$ is $c_2$.
According to the protocol, there are two possible cases.

In the first case, $c_1.view = c_2.view$ and the base log of $v'$ has a bigger \var{append-slot}.
In this case, we know that $c$ is included in the base log according to \autoref{claim:log_inclusion}.

In the second case, $c_1.view < c_2.view$.
Let's denote $c_2.view$ as $v''$, so $v'' < v'$.
According to \autoref{claim:monolithic_view}, $v = c.view \leq c_1.view$.
So, $v \leq v'' < v'$.
Also, $v''.init$ exists in the base log of $v'$ according to \autoref{claim:initview_first}.
Given \autoref{claim:log_matching}, the base log of $v'$ and the base log of $v''$ are at least identical up to $v''.init$.
Given that $v'$ is the smallest view that violates the claim, the base log of $v''$ includes $c$.
So, the base log of $v'$ also includes $c$.
This leads to a contradiction given the definition of the base log.
\end{proof}

\begin{claim}\label{claim:initiator_completeness}
\emph{Initiator Completeness Property:}  If $c$ is committed in $v$, then $c$ is included in the log of the initiator of any view $v'$ that satisfies $v' > v$.
\end{claim}

\begin{proof}
\autoref{claim:initiator_completeness_lemma} proves the case when $c.view = v$.
For $c.view \neq v$, according to the protocol, the commit happens when $v.init$ is committed in $v$.
Also, we know that $c.view < v$ from \autoref{claim:propose_then_add}.
So, for $v.init$, we know that it exists in the log of the initiator of any view bigger than $v$, which includes $v'$. 
According to \autoref{claim:log_matching}, we know that $c$ is also included in the log of the initiator of $v'$.
\end{proof}

\begin{claim}\label{claim:consistency_induction_block}
If command $c$ is committed at slot $s$ in view $v$, for any $v'$ that satisfies $v' \ge v$, any command $c'$ that is committed at $s$ in $v'$ is $c$.
\end{claim}

\begin{proof}
If $v = v'$, according to \autoref{claim:slot_view_unique_append},  $c' = c$.
Otherwise, $c'$ is committed when $v'.init$ is committed, and $c'$ is included in the base log of $v'$.
According to \autoref{claim:initiator_completeness} and \autoref{claim:slot_view_unique_append}, $c' = c$.
\end{proof}

Using induction, \autoref{claim:consistency_induction_block} leads to \autoref{theorem:safety}.


\end{document}